\documentclass[11pt,letterpaper]{article}
\usepackage[left=1in,right=1in,top=1in,bottom=1in]{geometry}
\usepackage{amsmath,amsthm,amssymb}
\usepackage{hyperref}
\usepackage{url}
\usepackage{enumitem}
\usepackage{moreenum}

\newcommand{\cc}[1]{\mathsf{#1}}
\newcommand{\set}[1]{\{#1\}}
\newcommand{\dash}{\text{-}}
\DeclareMathOperator{\dom}{dom}
\DeclareMathOperator{\img}{img}
\newcommand{\N}{\mathbb{N}}

\DeclareMathOperator{\poly}{poly}

\theoremstyle{plain}
\newtheorem{proposition}{Proposition}[section]
\newtheorem{theorem}[proposition]{Theorem}
\newtheorem{corollary}[proposition]{Corollary}
\newtheorem{question}[proposition]{Open Question}
\newtheorem{lemma}[proposition]{Lemma}
\newtheorem{observation}[proposition]{Observation}

\theoremstyle{definition} 
\newtheorem{definition}[proposition]{Definition}
\newtheorem{conjecture}[proposition]{Conjecture}
\newtheorem{remark}[proposition]{Remark}
\newtheorem{example}[proposition]{Example}
\newtheorem{cform}{Classical AC}
\newtheorem{pform}{Polynomial-Time AC (PAC)}

\newtheorem{cresult}{Classical Result}
\newtheorem{panalogue}{Polynomial-Time Analogue}

\begin{document}

\title{Polynomial-Time Axioms of Choice and \\
Polynomial-Time Cardinality \\ 
{\normalsize \textit{Dedicated to the memory of Alan L. Selman}}}
\author{Joshua A. Grochow}

\maketitle

\pagestyle{myheadings}
\markboth{Polynomial-Time Axioms of Choice - Joshua A. Grochow}{Polynomial-Time Axioms of Choice - Joshua A. Grochow}

\begin{abstract}
There is no single canonical polynomial-time version of the Axiom of Choice (AC); several statements of AC that are equivalent in Zermelo-Fraenkel (ZF) set theory are already inequivalent from a constructive point of view, 
and are similarly inequivalent from a complexity-theoretic point of view. In this paper we show that many classical formulations of AC, when restricted to polynomial time in natural ways, are equivalent to standard complexity-theoretic hypotheses, including several that were of interest to Selman.  This provides a unified view of these hypotheses, and we hope provides additional motivation for studying some of the lesser-known hypotheses that appear here.

Additionally, because several classical forms of AC are formulated in terms of cardinals, we develop a theory of polynomial-time cardinality. Nerode \& Remmel (\emph{Contemp. Math.} 106, 1990 and Springer Lec. Notes Math. 1432, 1990) developed a related theory, but restricted to unary sets. Downey (Math. Reviews MR1071525) suggested that such a theory over larger alphabets could have interesting connections to more standard complexity questions, and we illustrate some of those connections here.

The connections between AC, cardinality, and complexity questions also allow us to highlight some of Selman's work. We hope this paper is more of a beginning than an end, introducing new concepts and raising many new questions, ripe for further research.
\end{abstract}

\section{Introduction} \label{sec:intro}
\begin{quotation}
\noindent ``The Axiom of Choice is obviously true, the well-ordering principle is obviously false, and who can tell about Zorn's Lemma?'' ---Jerry L. Bona, \cite[p.~145]{Handbook}
\end{quotation}

The Axiom of Choice (AC), in one of its many equivalent incarnations, says that for every collection $(X_i)_{i \in I}$ of non-empty sets $X_i$ (where $I$ is an arbitrary index set), there is a ``choice function'' $x\colon I \to \bigcup_{i \in I} X_i$ such that $x(i) \in X_i$ for all $i$. Despite being non-constructive, this axiom is extremely useful throughout much of mathematics. In classical, Zermelo--Fraenkel set theory (ZF), there are many statements that are equivalent to AC. We refer to the books by Herrlich \cite{Herrlich} and Rubin \& Rubin \cite{RubinRubin, RubinRubin2} for the classical AC and its many equivalent versions. Yet despite being equivalent in ZF, many versions often ``feel'' different, as captured nicely by Bona in the above quotation.

In this paper, we study several statements that are classically equivalent to AC and propose polynomial-time analogues that seem not to be equivalent to one another (under standard complexity assumptions). We show their relationship with some standard---and some less standard---complexity questions, most with relations to Alan Selman's work, which is part of why we thought it fitting to submit to this memorial volume. 

\begin{quotation}
\noindent \textbf{Informal Definition.} A \emph{polynomial-time axiom of choice} is any statement $S$ which, if the words ``polynomial-time,'' ``polynomially-bounded,'' and similar were removed from $S$, would be equivalent (in Zermelo-Fraenkel set theory) to the usual Axiom of Choice. 
\end{quotation}

Throughout this paper, our general philosophy is that ``set'' should become ``language in $\cc{P}$'' and ``function'' should become ``function in $\cc{FP}$.''  Where there are several obvious choices for how to formulate a polynomial-time version of AC equivalent to a given standard formulation, we generally avoid formulations that are trivially true or trivially false.

Because several statements classically equivalent to AC are phrased in terms of cardinals, we also develop a polynomial-time theory of ``p-cardinals'' in Section~\ref{sec:cardinality}. The basic idea is as follows: Cantor's idea of cardinality was that two sets have the same ``size'' iff there is a bijection between them, so by analogy, we say two languages $A \subseteq \Sigma^*, B \subseteq \Gamma^*$ should have the same p-cardinality if there is a polynomial-time computable, polynomial-time invertible bijection between $A$ and $B$. The key difference between this and the notion of p-isomorphism is that a p-isomorphism $A \to B$ must be given by a \emph{total} p-computable, p-invertible bijection $f\colon \Sigma^* \to \Gamma^*$ such that $f(A)=B$, whereas in our case we focus more on the sets themselves, only requiring our bijection to be computable and invertible in polynomial time \emph{when restricted to the sets} $A,B$. Outside these sets it need not be injective nor surjective.

Our notion of p-cardinality turns out to be related to several complexity-theoretic notions of interest in Selman's work, including p-enumerability (Observation~\ref{obs:enum} and Proposition~\ref{prop:enum}), the Isomorphism Conjecture (Section~\ref{sec:iso}), p-rankability (Section~\ref{sec:countable}), and immunity (in Propositions~\ref{prop:cardvsiso} and \ref{prop:NPimmune}). We suspect it has further relationships with other of Selman's interests including disjoint $\cc{NP}$ pairs \cite{GSSZdisjoint, GSTWdisjoint, GHSWdisjoint, GSSdisjoint}, mitoticity \cite{GPSZmitosis, GOPS07, GNSW17}, splitting \cite{GPSZSplitting}, and p-selectivity \cite{SelmanPsel, Selman81Psel, NaikSelmanPsel, HT}, but leave these as exciting future directions.

\subsection{Related Work}
Constructive versions of AC, even in the context of computability, have been studied previously, e.\,g., \cite{Carl, DM}, but we are not aware of other work studying polynomial-time analogues of AC.

For p-cardinality, the most closely related work is that of Nerode \& Remmel \cite{NerodeRemmelPET, NerodeRemmelIsol, NerodeRemmel1, NerodeRemmel2}. They developed not only a theory of polynomial-time cardinals, called ``polynomial equivalence types'' or ``PETs'' (in analogy with previously studied recursive equivalence types \cite{DekkerMyhill}, see \cite{Dekker}), but also began to develop a theory of polynomial-time structures such as vector spaces, groups, etc., in analogy with work on computable structures (see, e.\,g., \cite{CMT, AshKnight, MontalbanBook}). However, their polynomial-time theory was developed for unary languages. 

While the choice to use unary languages enabled many analogies with the recursive case to go through and led to a rich and interesting theory, it was mostly disconnected from more standard complexity questions such as $\cc{P}$ versus $\cc{NP}$, because the latter are typically formulated over an alphabet of size at least 2. Downey, in his review \cite{DowneyReview} of \cite{NerodeRemmelIsol} writes:
\begin{quotation}
``What is not yet clear is the relationship, if any, of complexity questions (such as $\cc{P} \stackrel{?}{=} \cc{NP}$) with \cite{NerodeRemmelIsol}. For instance, it would be truly fascinating if the techniques could be used to prove $\cc{P} = \cc{NP}$ to be independent of [intuitionistic Zermelo--Fraenkel set theory] via an analogue of McCarty's work \cite{McCarty, McCartyThesis}. Perhaps what is needed is the development of the theory over the standard languages $\subseteq \{0,1\}^*$.'' --R. Downey \cite{DowneyReview}
\end{quotation}
Here (Section~\ref{sec:cardinality}) we pursue one such development of the theory over standard languages, that is, over alphabets of size at least 2. While we do not yet realize Downey's suggestion about independence from IZF, we lay some possible foundations of such a theory. In fact, this paper originally was focused on AC, with p-cardinality as just a side note, and it was our discovery of Downey's quote that spurred us on to develop some of the foundations of p-cardinality further.

\subsection{Selman's influence on my work}
I only met Alan personally a few times---I think once when I visited Steve Homer at Boston University (to give my first talk on complexity theory research), and once or twice at the Conference on Computational Complexity. But even from those few meetings one could tell that he was nice, generous with his ideas and time, had broad interests, and seemed genuinely very happy to talk about complexity theory. 

But his influence on my work extends far beyond just our few chance meetings. The general theme is this: when I started graduate school and was talking with Lance Fortnow and other members of the Theory group at U. Chicago (at the time: Stuart Kurtz, Ketan Mulmuley, Janos Simon, Laci Babai, and shortly after I joined, also Sasha Razborov), I was just over-the-moon with structural complexity theory. Properties of complexity classes, properties of (formal) languages, p-isomorphism, complexity cores, p-selectivity---I loved it all (still do!). Early on in my graduate career Selman's two edited volumes \cite{CTR, CTR2} were very influential on my thinking. Every chapter of \cite{CTR} is related to something I have worked on at one point in my life (even if not resulting in any publications). They are still some of my favorite advanced complexity books on my shelf---and yes, I actually have the physical books, on a physical bookshelf.

While working on my first complexity paper \cite{FortnowGrochow}, Lance pointed me to Selman's taxonomy of function classes \cite{SelmanFunctions, Selman93Functions, Selman96Functions}, and one can see that $\cc{NPMV}$ and its kin play an important role in that paper. Personal admission: I never really understood the definition of $\cc{TFNP}$ until reading Selman's work on $\cc{NPMV}, \cc{NPSV}$, and their relatives. One of Selman's results was in fact part of the motivation for \cite{FortnowGrochow} (discussed in my master's thesis \cite[Sec.~3.2]{GrochowMasters}): although most function classes behave like their decision class counterparts, we have $\cc{P}^{\cc{NP}[log]} = \cc{P}_{tt}^{\cc{NP}}$ \cite{hemachandraPhD, wagner, bussHay} but Selman showed that for the corresponding function classes this is unlikely:

\begin{theorem}[{Selman \cite{SelmanFunctions}}]
$\cc{FP}^{\cc{NP}[log]} = \cc{FP}_{tt}^{\cc{NP}}$ implies $\cc{NP} = \cc{RP}$ and $\cc{P} = \cc{UP}$.
\end{theorem}

After that first complexity paper, I was hooked (if I wasn't before). One of my favorite open oracle questions is still from Selman's work. It remains open whether $\cc{NP} = \cc{UP}$ collapses $\cc{PH}$. But more than that, we don't even know that it requires non-relativizing techniques. We wish to highlight this question here:

\begin{question}[Folklore] \label{q:NPUP}
Build an oracle relative to which $\cc{NP} = \cc{UP}$ but $\cc{PH}$ is infinite.
\end{question}

Hemaspaandra, Naik, Ogihara, and Selman \cite{HNOS96} showed that a closely related statement does in fact collapse $\cc{PH}$ to $\cc{\Sigma_2 P}$. Namely, they showed that this collapse of $\cc{PH}$ follows from the assumption that $\cc{NPMV}_g \subseteq_c \cc{NPSV}_g$. All terms are defined in Section~\ref{sec:prelim}, but for now, we can say that the latter statement is equivalent to: for every language $A \subseteq \Sigma^* \times \Sigma^*$ in $\cc{P}$, there is another language $B$ ``refining'' $A$, in the sense that for all $x \in \Sigma^*$, if $(x,y) \in A$, then there exists some $y_x$ such that $(x,y_x) \in A$ and $(x,y_x)$ is the unique element of $B$ with first coordinate $x$. This implies $\cc{NP}=\cc{UP}$, and feels philosophically very closely related (indeed, it is the kind of statement that  someone not familiar with it might mistake for the definition of $\cc{NP}=\cc{UP}$). 
I spent at least a summer, possibly more, in graduate school under Lance's guidance trying to answer Question~\ref{q:NPUP}. The issue is that most naive ways of thinking about how to get $\cc{NP}=\cc{UP}$---and I was definitely naive at the time, probably still am---also imply $\cc{NPMV}_g \subseteq_c \cc{NPSV}_g$, and thus that $\cc{PH} = \cc{\Sigma_2 P}$. To get around this barrier, we first sought to build an oracle relative to which $\cc{NP} = \cc{UP}$ but $\cc{PH} \neq \cc{\Sigma_2 P}$. We did not manage to do that, but I mention this because it led us to another of Selman's works. Rather than trying to fully separate the levels of $\cc{PH}$, which seems to require using the full strength of $\cc{AC}^0$ circuit lower bounds, we just sought to separate its second level, and Baker and Selman had built an oracle separating just two levels of $\cc{PH}$ \cite{BakerSelman}, so we looked to their work for inspiration. 

At this point I don't remember why, but around the same time I was also led to read Selman's work on disjoint pairs, mitoticity, immunity, and p-selectivity (e.\,g., \cite{GOPS07, GSSdisjoint, GS, HNOSPsel, NaikSelmanPsel}). In fact, Selman's work on search-to-decision reductions was part of what spurred my interest, leading me to discuss them in the context of Group Isomorphism with Youming Qiao in graduate school over a decade ago; we just recently made progress on that problem \cite{GQp}. 

Even though we didn't solve Question~\ref{q:NPUP}, it was a glorious several months being fully immersed in Selman's work, that had an influence on my thinking and my career, and will stick with me for the rest of my life. Thanks Alan---your intellectual legacy lives on, and you are missed.

\section{Preliminaries} \label{sec:prelim}
Throughout, $\Sigma$ (and sometimes also $\Gamma$) is a finite alphabet of size at least 2, containing at least the symbols $0,1$. $\epsilon$ denotes the empty string, the unique string of length zero. $|x|$ denotes the length of a string. Length-lexicographic ordering is defined by $x <_{lex} y$ if $|x| < |y|$, or if $|x|=|y|$ and $x$ lexicographically precedes $y$, where lexicographic ordering is given according to some fixed (but unspecified, implicit) total order on $\Sigma$. $\Sigma^*$ denotes the set of all finite strings over the alphabet $\Sigma$. 

We move back and forth freely between $\Sigma^*$ and $\N$ as follows. The natural number associated to $x \in \Sigma^*$ is the number of strings that are $\leq_{lex} x$, and vice versa, so e.\,g., if $\Sigma = \{0,1,2\}$ we have the association

\begin{tabular}{ccccccccccccccccccccc}
0 & 1 & 2 & 3 & 4 & 5 & 6 & 7 & 8 & 9 & 10 & 11 & 12 & 13 & 14 & 15 & 16 & 17 & $\dotsb$ \\
$\epsilon$ & 0 & 1 & 2 & 00 & 01 & 02 & 10 & 11 & 12 & 20 & 21 & 22 & 000 & 001 & 002 & 010 & 011 & $\dotsb$
\end{tabular}

\noindent For two strings $x,y$ and a natural number $n$, when we write $x+y$ or $x+n$, we first convert all strings to natural numbers using the above convention, then add them, and then convert back to the corresponding string, if needed.

For a language $A \subseteq \Sigma^*$ we abuse notation by writing $A(x)=1$ iff $x \in A$ and $A(x)=0$ iff $x \notin A$. $\overline{A}$ denotes the complement of $A$, $\overline{A} = \{x \in \Sigma^* : x \notin A\}$. 

A polynomial-time many-one reduction from $A$ to $B$ is a p-computable total function $f$ such that $x \in A \Leftrightarrow f(x) \in B$ for all $x \in \Sigma^*$, equivalently, $A(x) = B(f(x))$. In this case we write $A \leq_m^p B$. If $A \leq_m^p B$ and $B \leq_m^p A$ we write $A \equiv_m^p B$ and say that $A$ and $B$ have the same polynomial-time many-one degree (of difficulty). A bounded truth-table, or $\leq_{btt}^p$, reduction is a nonadaptive Turing reduction that makes $O(1)$ queries (we only use $\leq_{btt}^p$ in reference to other people's results). $A \subseteq \Sigma^*$ and $B \subseteq \Gamma^*$ are p-isomorphic, denoted $A \cong^p B$, if there is a total, polynomial-time computable bijection $f\colon \Sigma^* \to \Gamma^*$ for which $f^{-1}$ is also computable in polynomial time, and such that $f(A)=B$.

For two languages $A, B \subseteq \Sigma^*$, $A \oplus B$ denotes the ``tagged'' disjoint union $\{0a : a \in A \} \cup \{1b : b \in B\}$ and $A \times B$ denotes the usual Cartesian product. Pairs, tuples, etc. of strings are encoded using any of the standard polynomial-time computable and polynomial-time invertible bijections $\Sigma^* \to \Sigma^* \times \Sigma^*$. We write $A=^*B$ if their symmetric difference $(A\backslash B) \cup (B \backslash A)$ is finite.

We follow notation for partial, multi-valued functions set down by Selman \cite{SelmanFunctions, HNOS96}.  A (possibly) partial, multivalued function $f$ is technically a relation $\subseteq \Sigma^* \times \Sigma^*$, but we think of the first coordinate of this relation as inputs and the second as outputs, and we use standard functional language and notation. Rather than the relation notation $x f y$, we write $f(x) \mapsto y$, but note that there may be more than one such $y$ for a given $x$ (so we are careful \emph{not} to write $f(x) = y$ unless there is exactly one output for $f(x)$). We write 
\[
set\dash f(x) = \{ y : f(x) \mapsto y\} \qquad\text{and}\qquad \dom(f) = \{x : set\dash f(x) \neq \emptyset \}.
\]

Nondeterministic Turing machine transducers---machines with separate output tapes---naturally compute partial, multi-valued functions. Given such a machine $M$, we define $f_M(x) \mapsto y$ iff there is some accepting computation of $M$ that halts with $y$ on its output tape, and we say that $f_M$ is computed by $M$. There is a related but subtly different function also associated with $M$ that will be useful, namely $acc_M(x) \mapsto y$ iff $y$ is the ordered list of nondeterministic choices made on an accepting branch of $M$; note that $\dom(f_M) = \dom(acc_M) = L(M)$, the language accepted by $M$.

Abusing ordinary function notation, for two partial, multi-valued functions $f,g$, we may write, e.g., $f(x,g(x))$ to indicate the partial, multivalued function $c(x)$ with $c(x) \mapsto z$ iff there is an output $g(x) \mapsto y$ such that $f(x,y) \mapsto z$. (This also corresponds to composition of relations.)

We refer to a nondeterministic, polynomial-time machine $M$ as ``an NP machine''. 

$\cc{NPMV}$ is the class of (possibly) partial, multi-valued functions computed by NP machines. $\cc{NPSV}$ is the subclass of $\cc{NPMV}$ consisting of those functions that are in fact single-valued, though they may still be partial and may have multiple accepting branches (as long as they all make the same output). $\cc{PF}$ is the class of functions computed by NP machines that use no nondeterminism, but may still be partial (viz., if they reject some input). $\cc{FP}$ is the class of total functions computed by deterministic polynomial-time machines.

A multivalued function $f$ is a \emph{refinement} of a multivalued function $g$ if $\dom(f) = \dom(g)$ and $set\dash f(x) \subseteq set\dash g(x)$ for all $x$ (this second condition is equivalent to $f(x) \mapsto y$ implying $g(x) \mapsto y$). If $\mathcal{F}_1, \mathcal{F}_2$ are two classes of functions, we write $\mathcal{F}_1 \subseteq_c \mathcal{F}_2$ if every $f \in \mathcal{F}_1$ has a refinement in $\mathcal{F}_2$.

For any function $f$---whether it be injective or not, single-valued or not, total or not, surjective or not---we define $f^{-1}$ to be the partial multivalued function defined by $f^{-1}(y) \mapsto x$ iff $f(x) \mapsto y$. Given a class $\mathcal{F}$ of functions, we say $f$ is invertible in $\mathcal{F}$ if $f^{-1}$ has a refinement in $\mathcal{F}$. By our symbolic conventions, note that even though $f^{-1}$ may be partial and multivalued, we always have $f^{-1}(f(x)) \mapsto x$. We also have $f(f^{-1}(x)) \mapsto x$ iff $x$ is in the image of $f$.

A function $f$ is (polynomially) \emph{honest} if there is a polynomial $p$ such that for every input $x$, there is some output $y$ such that $|x| \leq p(|y|)$. Note that honesty is a necessary condition to have some refinement of $f^{-1}$ be computable in polynomial space, let alone polynomial time.

For these function classes, there are a number of subscripts to use that denote subclasses:
\begin{itemize}
\item The subscript $t$ denotes the subclass of total functions, that is, whose domain is all of $\Sigma^*$

\item The subscript $h$ denotes the subclass of honest functions

\item The subscript $g$ denotes the subclass of functions whose \emph{graph}, $\set{ (x,y) : f(x) \mapsto y}$, is decidable in polynomial time.
\end{itemize}
These subscripts may be used in any combination, e.\,g. $\cc{NPMV}_{ght}$. 

Selman \cite{SelmanFunctions} and Hemaspaandra, Naik, Ogihara, and Selman \cite{HNOS96} showed that 
\[
\cc{NPSV}_t = \cc{FP}^{\cc{NP} \cap \cc{coNP}}. 
\]

\noindent \textbf{Theorem/Definition Q.} (\cite[Thm.~2]{FFNRPropQ}) Hypothesis Q is any of the following equivalent statements
\begin{enumerate}[label=Q\arabic*]
\item For all NP machines $M$ that accept $\Sigma^*$, there exists a polynomial-time computable function $g_M$ such that for all $x$, $g_M(x)$ outputs an accepting computation of $M$ on $x$.

\item All polynomial-time computable onto honest functions are invertible in $\cc{PF}$.

\item \label{Q:NPMV} $\cc{NPMV}_{t} \subseteq_c \cc{FP}$.

\item For all $S \in \cc{P}$ such that $S \subseteq \text{SAT}$, there exists a polynomial-time computable $g$ such that for all $x \in S$, $g(x)$ outputs a satisfying assignment of $x$.

\item $\cc{P} = \cc{NP} \cap \cc{coNP}$ and $\cc{NPMV}_t \subseteq_c \cc{NPSV}_t$.

\item For all NP machines $M$ such that $L(M) = \text{SAT}$, there is $f_M \in \cc{PF}$ such that for all $\varphi \in \text{SAT}$, $\varphi(f_M(\varphi, acc_M(\varphi)))=1$, that is, $f_M$ takes in $\varphi$ and any accepting computation of $M$ on $\varphi$, and outputs at least one satisfying assignment to $\varphi$, and all such outputs satisfy $\varphi$.

\item For all NP machines $M,N$ such that $L(M) \subseteq L(N)$, there is $f_M \in \cc{PF}$ such that for all $x \in L(M)$, $f_M(x, acc_M(x))$ makes some output, and all its outputs are accepting computations of $N(x)$.

\item For all $L \in \cc{P}$ and all NP machines $M$ accepting $L$, there is $f_M \in \cc{PF}$ such that $f_M(x)$ is an accepting computation of $M(x)$.
\end{enumerate}

\section{Polynomial-time Axioms of Choice} \label{sec:AC}
The formulations in this section are mostly based on choice functions for collections of sets.  If $S = \set{S_{i} : i \in I}$ is a collection of sets (where $I$ is an arbitrary index set), then a choice function for $S$ is a function $f$ such that $f(i) \in S_{i}$ for all $i \in I$.

\begin{definition}[Collection of languages, honestly non-empty, choice function]
Given a complexity class $\mathcal{C}$, a \emph{$\mathcal{C}$ collection of languages} is a single language $L \in \mathcal{C}$ such that each $L_x := \set{y : (x,y) \in L}$ is also in $\mathcal{C}$. 

A collection $L$ of languages is \emph{honestly non-empty} if there are polynomials $p,q$ of positive degree such that, for every $x$, $L_x$ contains at least one string $y$ of length $p(|x|) \leq |y| \leq q(|x|)$. We call $p,q$ ``honesty (lower, resp. upper) bounds'' for $L$.

A partial, multi-valued function $f$ is a \emph{choice function} for a collection $L$ of languages if: (1) for all $x$, $set\dash f(x) \subseteq L_x$, and (2) for every nonempty $L_x$, $set\dash f(x)$ is nonempty.
\end{definition}

\begin{remark}
The polynomial upper bound in the definition of honestly non-empty is easy enough to justify at this point: there should at least exist a $y$ in $L_x$ that did not require super-polynomial-space to write down (so, e.g., if $L \in \cc{PSPACE}$, then such a $y$ could be found by a polynomial-space machine). But the polynomial lower bound may seem less intuitive. It turned out to be a very natural and useful condition at several points, e.g. Observation~\ref{obs:DP} and Proposition~\ref{prop:npac2}(1).
\end{remark}

\begin{cform}[{See \cite[AC1, p.~5]{RubinRubin}}] Every collection of non-empty sets has a choice function. \end{cform}

\begin{pform} \label{PAC:choice} Every $\cc{P}$ collection of honestly non-empty languages has a choice function in $\cc{FP}$. \end{pform}

\begin{proposition} 
PAC\ref{PAC:choice} holds iff $\cc{NPMV}_{gt} \subseteq_c \cc{FP}$.  
\end{proposition}

\begin{proof} 
($\Rightarrow$) Suppose PAC\ref{PAC:choice} holds. Let $f$ be an $\cc{NPMV}_{gt}$ function and let $L$ be the graph of $f$.  By definition of $\cc{NPMV}_{g}$, the graph of $f$ is in $\cc{P}$, so $L \in \cc{P}$.  Since $f$ is total, every $L_x$ is non-empty.  Since each branch of $f$ is polynomial-time, every output $y$ of $f(x)$ has $|y| \leq \poly(|x|)$. To get a polynomial lower bound as well, we modify $f$ to $\hat{f}$, defined by $\hat{f}(x) \mapsto (x,y) \Leftrightarrow f(x) \mapsto y$, so that every output now has $|x| \leq |\hat{f}(x)|$. Let $\hat{L}$ be the graph of $\hat{f}$; $\hat{L}$ is now a $\cc{P}$ collection of honestly non-empty languages. Thus, by assumption, there is a polynomial-time choice function $\hat{g}$ for $\hat{L}$. By construction of $\hat{L}$, $\hat{g}(x)$ always has the form $(x,y)$. Finally, define $g(x) = y \Leftrightarrow \hat{g}(x) = (x,y)$. It is clear that $g$ is also in $\cc{FP}$ and refines $f$.

($\Leftarrow$) Conversely, suppose $\cc{NPMV}_{gt} \subseteq_{c} \cc{FP}$, and let $L$ be a $\cc{P}$ collection of honestly non-empty languages, with honesty bounds $p,q$.  Define $f$ by
\[
set\dash f(x) = \set{y : p(|x|) \leq |y| \leq q(|x|) \text{ and } (x,y) \in L}
\]
Since $L \in \cc{P}$, the graph of $f$ is in $\cc{P}$: use the polynomial-time decider for $L$, together with verifying the length relationship between $x,y$.  Moreover, since $L_{x}$ is honestly non-empty for every $x$, $f$ is total, so $f$ is in $\cc{NPMV}_{gt}$.  By assumption $f$ has an $\cc{FP}$ refinement, which is clearly a polynomial-time choice function for $L$.
\end{proof}

The same argument can be trivially modified to several other standard complexity hypotheses, by varying the constraints in PAC\ref{PAC:choice}. We give two such here:

\begin{proposition} \label{prop:npac1}
Every $\cc{NP}$ collection of honestly non-empty languages has...
\begin{enumerate}
\item ...an $\cc{NPSV}_t$ choice function iff $\cc{NPMV}_t \subseteq_c \cc{NPSV}_t$.
\item ...an $\cc{FP}$ choice function iff $\cc{NPMV}_t \subseteq_c \cc{FP}$; note the latter is Hypothesis~\ref{Q:NPMV}.
\end{enumerate}
\end{proposition}

Thus we see that Hypothesis Q is equivalent to a mixed nondeterministic/deterministic polynomial-time version of the Axiom of Choice. We will see below that Polynomial-Time AC \ref{pac:Q} is also equivalent to Q.

\begin{proof}
Let $\mathcal{F} \in \set{\cc{NPSV}_t, \cc{FP}}$. 

($\Rightarrow$) Suppose every $\cc{NP}$ collection of honestly non-empty languages has an $\mathcal{F}$ choice function. Let $f$ be in $\cc{NPMV}_{t}$. As before, let $\hat{f}(x) \mapsto (x,y) \Leftrightarrow f(x) \mapsto y$, and let $\hat{L}$ be the graph of $\hat{f}$. Since every $\cc{NPMV}$ function has its graph in $\cc{NP}$, $\hat{L}$ is in $\cc{NP}$. As above, $\hat{L}$ is an honestly non-empty collection of languages. By assumption, there is a choice function $\hat{g} \in \mathcal{F}$ for $\hat{L}$. Define $g(x)$ by $g(x) = y \Leftrightarrow \hat{g}(x) = (x,y)$; for either choice of $\mathcal{F}$ it is clear that $\hat{g} \in \mathcal{F} \Leftrightarrow g \in \mathcal{F}$. As above, $g$ is clearly a refinement of $f$, and thus $\cc{NPMV}_t \subseteq_c \mathcal{F}$.

($\Leftarrow$) Conversely, suppose $\cc{NPMV_{t}} \subseteq_{c} \mathcal{F}$, and let $L$ be an $\cc{NP}$ collection of honestly non-empty languages, with honesty bounds $p,q$.  Define $f$ by
\[
set\dash f(x) = \set{y : p(|x|) \leq |y| \leq q(|x|) \text{ and } (x,y) \in L}.
\]
Since $L \in \cc{NP}$, the graph of $f$ is in $\cc{NP}$: use the polynomial-time verifier for $L$, together with verifying the length relationship between $x,y$.  Moreover, since $L_{x}$ is honestly non-empty for every $x$, $f$ is total, so $f$ is in $\cc{NPMV}_{t}$.  By assumption $f$ has an $\mathcal{F}$ refinement, which is clearly an $\mathcal{F}$ choice function for $L$.
\end{proof}

\begin{cform}[{See \cite[AC2, p.~5]{RubinRubin}, \cite[Prop.~2.1]{Herrlich}}] Given any set $X$ of \emph{pairwise disjoint} nonempty sets, there is a set $C$ that contains exactly one element from each set in $X$. \end{cform}

\begin{definition}[Transversal]
Given a collection $L$ of languages, a \emph{transversal} for $L$ is a language $C$ such that $C \cap L_x$ has exactly one element, for each $x$ such that $L_x$ is nonempty.
\end{definition}

\begin{pform} \label{pac:disjoint} Every polynomial-time collection of honestly nonempty languages that are pairwise disjoint, has a transversal in $\cc{P}$. \end{pform}

While we have not been able to get a polynomial-time upper bound on transversals of such collections, we can give an upper bound within the second level of $\cc{PH}$. More specifically, recall that $\cc{DP} = \{ L \subseteq \Sigma^* : L = L_1 \backslash L_2 \text{ for some } L_1, L_2 \in \cc{NP}\}$. $\cc{DP}$ is contained in $\cc{P}^\cc{NP}$, and even in $\cc{P}^{\cc{NP}[2]}$ (the machine only makes $2$ queries on inputs of length $n$). For any language $L$, we relativize $\cc{DP}$ to $\cc{DP}^L$, by taking the preceding definition and replacing $\cc{NP}$ with $\cc{NP}^L$.

\begin{observation} \label{obs:DP}
Every collection $L$ of pairwise disjoint, honestly non-empty languages has a transversal in $\cc{DP}^{L}$.
\end{observation}

\begin{proof}
Let $p,q$ be honesty bounds for $L$, so that each $L_x$ contains some $y$ with $p(|x|) \leq |y| \leq q(|x|)$. Let $C$ consist of the lexicographically first element of $L_x$ with length at least $p(|x|)$, for each $x \in \Sigma^*$. Since the $L_x$ are pairwise disjoint, $C$ does not contain more than one element of each $L_x$, and thus it contains exactly one and is a transversal. 

We now show that membership in $C$ can be decided in $\cc{DP}^{L}$. Let $L_1 = \{ y : (\exists x)[q^{-1}(|y|) \leq |x| \leq p^{-1}(|y|) \text{ and } (x,y) \in L\}$ and let $L_2 = \{ y : (\exists y' <_{lex} y)(\exists x)[ \max\{q^{-1}(|y'|), q^{-1}(|y|)\} \leq |x| \leq \min\{p^{-1}(|y|), p^{-1}(|y'|)\} \text{ and } (x,y) \in L \text{ and } (x,y') \in L\}$.
Note that $L_1, L_2$ are both in $\cc{NP}^{L}$, and we have $C = L_1 \backslash L_2$.
\end{proof}

Thus for collections in $\cc{P}$, there always exists an $\cc{DP}$ transversal.

\begin{question} \label{q:pairwise_disjoint} 
Is there a $\cc{P}$ collection of pairwise disjoint, honestly non-empty polynomial-time languages such that every transversal for $L$ is $\cc{NP}$-hard? $\cc{NP}$-complete? $\cc{DP}$-complete? \end{question}

While we do not have a clean equivalence between PAC\ref{pac:disjoint} and a more standard complexity hypothesis, we do have such a clean equivalence between more standard hypotheses and some $\cc{NP}$ versions of PAC\ref{pac:disjoint}.

\begin{proposition} \label{prop:npac2}
\begin{enumerate}
\item $\cc{NPMV}_{gt} \subseteq_c \cc{NPSV}_{t}$ 
iff 
\begin{quotation}
\noindent (*) every $\cc{P}$ collection of honestly nonempty, pairwise disjoint languages has a transversal in $\cc{NP}$.
\end{quotation}

\item $\cc{NPMV}_t \subseteq_c \cc{NPSV}_t$ iff 
\begin{quotation}
\noindent (*$^\prime$) every $\cc{NP}$ collection of honestly nonempty, pairwise disjoint languages has a transversal in $\cc{NP}$.
\end{quotation}
\end{enumerate}
\end{proposition}

\begin{proof}
(1) Suppose (*) holds, and let $f \in \cc{NPMV}_{gt}$. Let $L$ be $\{(x, (x,y)) : x \in \Sigma^*, f(x) \mapsto y\}$ be a modified version of the graph of $f$. By definition the graph of $f$ is in $\cc{P}$, so $L$ is also in $\cc{P}$. Since $f$ is total and each nondeterministic branch uses at most polynomial time, every $L_x$ is honestly nonempty.  Since every element of $L_x$ is of the form $(x,*)$, $L_x$ and $L_y$ are disjoint when $x \neq y$. Hence there is a set $C \in \cc{NP}$ that contains exactly one element from each $L_x$. We define an $\cc{NPSV}_{gt}$ function $g$ that, on input $x$, nondeterministically guesses $y$ and $w$, uses $w$ as a witness to verify that $(x,y) \in C$, and outputs $y$. This $g$ is an $\cc{NPSV}_t$ refinement of $f$. 

Conversely, suppose $\cc{NPMV}_{gt} \subseteq_c \cc{NPSV}_{t}$, and let $L$ be a $\cc{P}$ collection of honestly nonempty, pairwise disjoint languages with honesty bounds $p,q$.  Define $f$ by
\[
set\dash f(x) = \set{y : p(|x|) \leq |y| \leq q(|x|) \text{ and } (x,y) \in L}.
\]
Since $L \in \cc{P}$, the graph of $f$ is in $\cc{P}$.  Moreover, since $L_{x}$ is honestly non-empty for every $x$, $f$ is total.  Hence $f$ has an $\cc{NPSV}_{t}$ refinement $g$. Let $C$ be the image of $g$. It is clear that $C$ intersects each $L_x$ exactly once. The following nondeterministic machine shows that $C$ is in $\cc{NP}$: on input $y$, it guesses $x$ such that $p(|x|) \leq |y|$ (so in particular $|x| \leq \poly(|y|)$), guesses an accepting path of $g$, simulates the execution of $g(x)$ on that accepting path, and accepts if the output on that path is equal to $y$.

(2) We observe here how the preceding parts of the proof change for the second part. For $f \in \cc{NPMV}_t$, its graph (and modified graph as above) will be in $\cc{NP}$. (*$^\prime$) then allows us to conclude existence of a $C \in \cc{NP}$ as before, and the rest of the argument goes through. For the other direction of the argument, if $L$ is only in $\cc{NP}$, then $f$ will still be in $\cc{NPMV}_t$, so the assumption $\cc{NPMV}_t \subseteq_c \cc{NPSV}_t$ is still enough to get the $\cc{NPSV}_t$ refinement, and again the rest of the proof goes through.
\end{proof}

We note that it is only in the final step of the proof---in which a nondeterministic machine must guess a preimage $x \in g^{-1}(y)$---that nondeterminism seems to be needed, which is the only obstacle standing in the way of an equivalence with PAC\ref{pac:disjoint}. 

The preceding results let us show that two $\cc{NP}$ versions of AC are equivalent:

\begin{corollary}[``$\cc{NP}$ Axiom of Choice 1=$\cc{NP}$ Axiom of Choice 2'']
Every $\cc{NP}$ collection of honestly non-empty languages has an $\cc{NPSV}$ choice function iff every $\cc{NP}$ collection of pairwise disjoint honestly non-empty languages has an $\cc{NP}$ transversal. 
\end{corollary}

\begin{proof}
Proposition~\ref{prop:npac1}(1) says that the first condition in the statement here is equivalent to $\cc{NPMV}_t \subseteq_c \cc{NPSV}_t$, and Proposition~\ref{prop:npac2}(2) shows the latter is equivalent to the second condition in the statement here.
\end{proof}

\begin{observation}
PAC\ref{pac:disjoint} implies condition (*) from Proposition~\ref{prop:npac2}. 
\end{observation}

\begin{proof}
The conclusion of PAC\ref{pac:disjoint} implies existence of a $\cc{P}$ transversal, while the conclusion of (*) merely requires an $\cc{NP}$ transversal, but otherwise the two are the same.
\end{proof}

\begin{cform}[{See \cite[AC5, p.~5]{RubinRubin}}]
For every function $f$ there is a function $g$ such that $\dom(g)=\img(f)$ and for every $x \in \dom(g)$, $f(g(x))=x$.
\end{cform}

\begin{pform} \label{pac:owf}
For every partial polynomial-time honest function $f \in \cc{PF}$, there is a $g \in \cc{PF}$ with $\dom(g) = \img(f)$ such that for every $x \in \dom(g)$, $f(g(x)) = x$.
\end{pform}

Recall that a (worst-case) \emph{one-way function} is a one-one, honest function $f \in \cc{PF}$ such that $f^{-1} \notin \cc{PF}$, so PAC\ref{pac:owf} is simply the statement that one-way functions do not exist. This was one of the topics Selman worked on:

\begin{theorem}[{Grollman and Selman \cite{GS} and Ko \cite{Ko85}}]
PAC\ref{pac:owf} holds iff $\cc{P} = \cc{UP}$.
\end{theorem}

\section{Polynomial-time cardinality} \label{sec:cardinality}
\subsection{Definition and basic properties}

Several formulations of the Axiom of Choice are based on cardinalities. Here we develop polynomial-time analogues of these. This development is also partially motivated by the quotation from Downey \cite{DowneyReview} in Section~\ref{sec:intro}.

Let $\Sigma,\Gamma$ be two finite alphabets. For $A \subseteq \Sigma^*$ and $B \subseteq \Gamma^*$, by a ``partial function $f\colon A \to B$'' we mean a partial function $f\colon \Sigma^{*} \to \Gamma^{*}$ such that $A \subseteq \dom(f)$ and $f(A) \subseteq B$.  For inputs $x \notin A$, we impose no restrictions on $f(x)$---it might be in $B$, outside of $B$, or undefined. When we speak of properties of such a partial function---e.\,g., being computable in polynomial time, being injective, surjective, bijective, etc.---we only refer to these properties as they apply to inputs in $A$ and outputs in $B$.  For example, a partial function $f\colon A \to B$ is injective iff for all distinct $x,y \in A$, $f(x)$ and $f(y)$ are distinct.  Even if $f$ is defined on inputs outside of $A$, it may be non-injective on those inputs, or even have $f(x)=f(a)$ for some $x \notin A, a \in A$, and still be considered an injective partial function $A \to B$. We say a partial function $f\colon A \to B$ is polynomial-time invertible, or p-invertible, if $f^{-1}\colon f(A) \to A$ is single-valued (equivalently, $f$ is injective) and polynomial-time computable on $f(A)$. A partial function $f$ is \emph{length-increasing} if $|f(x)| > |x|$ for all $x \in \dom(f)$.

\newcommand{\card}[1]{{||#1||_p}}
\begin{definition} Two sets $A \subseteq \Sigma^*, B \subseteq \Gamma^{*}$ have the same \emph{p-cardinality} or are \emph{p-equipollent} if there are partial polynomial-time computable functions $f\colon A \to B$ and $g\colon B \to A$ such that $f \circ g = id_{B}$ and $g \circ f = id_{A}$. In this case we call $f$ a \emph{p-equipollence} $A \to B$.

We denote the p-cardinality of $A$---or, the same, its p-equipollence class---by $\card{A}$, so $\card{A}=\card{B}$ means that $A$ and $B$ are p-equipollent.

We say that the p-cardinality of $A$ is at most that of $B$, denoted $\card{A} \preceq \card{B}$, if $\card{A} = \card{B'}$ for some subset $B' \subseteq B$. $\card{A} \prec \card{B}$ denotes $\card{A} \preceq \card{B}$ and $\card{A} \neq \card{B}$.
 \end{definition}
 
The relation $\preceq$ between p-cardinals is transitive and reflexive, but we will see in Proposition~\ref{prop:not_poset} that it is \emph{not} a partial order, that is, there exist $A,B$ for which $\card{A} \preceq \card{B}$ and $\card{B} \preceq \card{A}$ but $\card{A} \neq \card{B}$; however see Theorem~\ref{thm:BH} for a partial positive result. This is why we maintain the notation $\preceq$ (rather than, say, $\leq$, as is typically done for cardinals). In Proposition~\ref{prop:not_total} we will also see that it is not total (there are incomparable p-cardinals). 

As we do with many-one reductions, it would make sense to introduce an equivalence relation $\card{A} \equiv \card{B}$ defined by $\card{A} \preceq \card{B}$ and $\card{B} \preceq \card{A}$; roughly speaking, the relationship between equivalence and equality of p-cardinals is analogous to the relationship between Karp equivalence and p-isomorphism of languages. The quotient relation $\leq$ on equivalence classes of p-cardinals is a partial order by construction, but Proposition~\ref{prop:not_total} still shows it to be only partial, not total. We will not use $\equiv$ much, though there is surely interesting theory to be explored there.

\emph{Notational note.} We have tried to consistently use the superscript$^p$ to denote a polynomial-time version of something, but in the case of cardinality, we wanted to use $||\cdot||$ as it is a standard symbol for cardinality, but wanted to avoid $||\cdot||^p$ because of its possible confusion with raising to the $p$-th power.

\begin{proposition} \label{prop:dtt}
For languages $A \subseteq \Sigma^*, B \subseteq \Gamma^*$, if $\card{A} = \card{B}$, then $A \equiv_m^p B$ unless $A = \Sigma^*$ or $B = \Gamma^*$. In the latter case, both of $A$ and $B$ are in $\cc{P}$.
\end{proposition}

\begin{proof}
Suppose that $\card{A} = \card{B}$, and first also suppose that $A \neq \Sigma^*$ and $B \neq \Gamma^*$. We will show that $B \leq_{m}^p A$; the reverse follows by symmetry. Let $a_0 \in \overline{A}$. 

Let $f\colon A \to B$ be a p-equipollence. Let $p$ be a polynomial and $M,N$ Turing machine transducers such that $M$ computes $f$ on inputs in $A$ in polynomial time and $N$ computes $f^{-1}$ on inputs in $B$ in time $\leq p(n)$. The following function $r$ is a polynomial-time reduction from $B$ to $A$: on input $x$, run $N(x)$ for $p(|x|)$ steps. If $N$ rejects or makes no output by that time, then $r(x) := a_0$. For $B \subseteq \dom (f^{-1}) \subseteq \dom(N)$. If $N(x)=y$ by the time it has made $p(|x|)$ steps, check that $M(y)=x$. If not, then output $a_0$ (for if $x$ were in $B$ then we have $M(N(x)) = f(f^{-1}(x)) = x$). If so, then let $r(x)=y$, for at this point $x \in B$ iff $r(x) \in A$. (Note that it is possible that $\dom(N)$ is strictly larger than $B$ and that $\dom(M)$ is strictly larger than $A$, and that $M(N(x)) = x$ for some $x$ outside $B$, so the final check that $r(x) \in A$ is needed.)

If $A = \Sigma^*$, then there is no $a_0$ to use, but in those cases the above reduction can directly decide whether the input $x$ was in $B$, so we get $B \in \cc{P}$. 
\end{proof}

Note that in the above proof we really used the fact that we had access to (partial) polynomial-time computable functions that were inverses of one another. If one tried to extend Nerode \& Remmel's theory \cite{NerodeRemmelPET, NerodeRemmelIsol} from unary languages to larger alphabets using only polynomial-time, bijective, honest functions, but without the stipulation of a polynomial-time inverse (even if one required such functions in both directions, but without requiring them to be inverses of one another) such a result seems much less likely to hold.

We note that the preceding result does \emph{not} extend in the natural way to $\card{A} \preceq \card{B}$:

\begin{proposition} \label{prop:notT}
There exist infinite, co-infinite sets $A,B$ such that $\card{A} \preceq \card{B}$ but $A \not\leq_T^p B$.
\end{proposition}

\begin{proof}
Let $A$ be an $\cc{EXP}$-complete set and $B = 0\Sigma^*$ (the prefix $0$ is just to ensure that $B$ is co-infinite). Then $0A = \{0a : a \in A\}$ is $\cc{EXP}$-complete, and $0A \subseteq B$, so $\card{A} = \card{0A} \preceq \card{B}$. But by the Time Hierarchy Theorem, $A \not\leq_T^p B$, for otherwise we would have $A \in \cc{P}$.
\end{proof}

We can use Proposition~\ref{prop:dtt} to show that $\preceq$ is \emph{not} a partial order:

\begin{proposition}[$\preceq$ is not a partial order on p-cardinals] \label{prop:not_poset}
There exist languages $A,B \subseteq \Sigma^*$ such that $\card{A} \preceq \card{B}$ and $\card{B} \preceq \card{A}$ (we might denote these together as $\card{A} \equiv \card{B}$), but $\card{A} \neq \card{B}$.
\end{proposition}

\begin{proof}
Let $A = 1\Sigma^*$, $B = 1(\Sigma^* \oplus B')$ for some $B' \notin \cc{P}$. First, we have $\card{A} = \card{10\Sigma^*}$, and $10\Sigma^* \subseteq 1(\Sigma^* \oplus B')$ (for any $B'$), so we get $\card{A} \preceq \card{B}$. In the opposite direction, since $B \subseteq A$, we get $\card{B} \preceq \card{A}$. But if $\card{A} = \card{B}$, Proposition~\ref{prop:dtt} would then give $A \equiv_{m}^p B$ (since neither is $\Sigma^*$), implying that $B \in \cc{P}$, contradicting the fact that $B \equiv_m^p B' \notin \cc{P}$. 
\end{proof}

\subsection{Comparison with p-isomorphism} \label{sec:iso}
Note that, a priori, having the same p-cardinality is a weaker condition than being p-isomorphic, since the bijections required for p-cardinality need not be total, nor be bijections on all of $\Sigma^{*}$. Our first two results in this direction help clarify the relationship between the two.

\begin{proposition}
If $A \cong^p B$, then $\card{A}=\card{B}$ and $\card{\overline{A}} = \card{\overline{B}}$. The converse holds if $A$ is in $\cc{P}$.
\end{proposition}

After the fact, we discovered that our argument for the converse here is very similar to, but more general than, one from Goldsmith, Hemachandra, and Kunen \cite{GHK}, reproduced below as Corollary~\ref{cor:GHK}(2).

\begin{proof}
Suppose $f\colon \Sigma^* \to \Gamma^*$ is a p-isomorphism $A \to B$. Then $f|_A \colon A \to B$ is a p-equipollence, and similarly $f|_{\overline{A}}\colon \overline{A} \to \overline{B}$ is a p-equipollence.

Suppose conversely that $f\colon A \to B$ is a p-equipollence (on $A$) and $g\colon \overline{A} \to \overline{B}$ is a p-equipollence, and furthermore that $A \in \cc{P}$. Then the following is a p-bijection $\phi\colon \Sigma^* \to \Gamma^*$ that is a p-isomorphism witnessing $A \cong^p B$. On input $x$, first we use the poly-time machine to decide whether $x \in A$. If so, then output $f(x)$, and if not, then output $g(x)$. Since $\Sigma^*$ is the disjoint union of $A$ and $\overline{A}$, $\phi$ is a total polynomial-time bijection. 

The following algorithm computes $\phi^{-1}$ in polynomial time. Let $t(n)$ be a polynomial upper bound on the running time of some algorithm computing $f^{-1}|_B$ and also of some algorithm computing $g^{-1}|_{\overline{B}}$. On input $x \in \Gamma^*$, try computing both $f^{-1}(x)$ and $g^{-1}(x)$ for $t(n)$ steps. If only one of them finishes in that amount of time, then we know whether $x$ was in $B$ or $\overline{B}$, and we use that as our output. However, note that the algorithm computing $f^{-1}$ may have domain larger than $B$, and the algorithm computing $g^{-1}$ may have domain larger than $\overline{B}$, so it is possible that both $f^{-1}(x)$ and $g^{-1}(x)$ output strings $y,z$, respectively, within $t(|x|)$ steps. If $y=z$ then the algorithm may unambiguously output $y$ without having determined whether $x \in B$ or not. If $y \neq z$, then we compute $\phi(y)$ and $\phi(z)$. Since $\phi$ is a bijection and $y \neq z$, we have $\phi(y) \neq \phi(z)$, so only one of them can be equal to $x$. (And at least one of them must equal $x$, since $x$ must be in at least one of $B, \overline{B}$.) If $\phi(y)=x$, then output $y$; if $\phi(z)=x$, then output $z$.
\end{proof}

\begin{proposition} \label{prop:cardvsiso}
\begin{enumerate}
\item There is a p-cardinal containing infinitely many distinct p-isomorphism classes.

\item There exist non-p-isomorphic sets $A,B$ of the same p-cardinality in which $\overline{A}, \overline{B}$ are both infinite and in $\cc{EXP}$.
\end{enumerate}
\end{proposition}

Note that there are no such examples where $A,B$ are finite, for in that case $\card{A}=\card{B}$ implies $A \cong^p B$.

\begin{proof}
1. For each $n \in \mathbb{N}$, let $S_n$ be the set of the lexicographically first $n$ strings. Define $A_n := \Sigma^* \backslash S_n$. Then all $A_n$ are p-equipollent to $A_0=\Sigma^*$: the p-equipollence $A_0 \to A_n$ simply shifts every string ``up $n$'' in terms of the corresponding $|\Sigma|$-ary numbers. But if $n \neq n'$, then $\overline{A_n}$ and $\overline{A_{n'}}$ are finite sets of different sizes, so they are not p-equipollent. Since $\card{\overline{A_n}} \neq \card{\overline{A_{n'}}}$ in this case, they cannot be p-isomorphic, hence neither can $A_n$ and $A_{n'}$.


2. We give an example with $\overline{A}, \overline{B}$ both infinite and in $\cc{EXP}$. Recall that a set $A$ is $\cc{P}$-immune if it is infinite but contains no infinite subset in $\cc{P}$. Let $A$ be any co-$\cc{P}$-immune set, that is, $\overline{A}$ is infinite but contains no infinite subsets in $\cc{P}$. For example, Ko \& Moore \cite{KM} showed that $\cc{P}$-bi-immune sets, in which both $A$ and its complement are $\cc{P}$-immune, exist in $\cc{EXP}$. Let $B = \{1x : x \in A\}$. Then a p-equipollence is given by $A \ni a \mapsto 1a \in B$, with inverse $B \ni 1x \mapsto x \in A$. But we claim that $A \not\cong^p B$. The short version is: $\cc{P}$-immunity is preserved by isomorphism, but $\overline{B}$ contains $0\Sigma^*$ as an infinite subset in $\cc{P}$. In more detail: suppose there is a polynomial-time computable bijection $f\colon \Sigma^* \to \Gamma^*$ such that $f(A)=B$ and with $f^{-1} \in \cc{FP}$ as well. Since $f^{-1}(\overline{B}) = \overline{A}$, we have that $f^{-1}(0\Sigma^*) \subseteq \overline{A}$. But $f^{-1}(0\Sigma^*)$ is in $\cc{P}$, for $x \in f^{-1}(0\Sigma^*)$ iff $f(x)$ begins with a 0, which is easily checked. As $0\Sigma^*$ is infinite and $f^{-1}$ is a bijection, this contradicts the $\cc{P}$-immunity of $\overline{A}$.
\end{proof}

\begin{observation} \label{obs:iso}
Let $L \notin \cc{P}$. If $L \equiv_m^p L'$ implies $L \cong^p L'$, then the p-isomorphism class of $L$ is equal to its p-equipollence class.
\end{observation}

\begin{proof}
Since $L \notin \cc{P}$, $\Sigma^* \not\equiv_m^p L$, so the p-isomorphism class of $L$ does not include $\Sigma^*$. Since p-isomorphism implies p-equipollence, and p-equipollence (avoiding $\Sigma^*$) implies $\equiv_m^p$ (Proposition~\ref{prop:dtt}), the result follows.
\end{proof}

\begin{corollary}[{cf. Kurtz, Mahaney, Royer \cite{KMRCollapse}}]
For every set $A$, there is a set $B \geq_m^p A$ such that the p-isomorphism class of $B$ is equal to its p-equipollence class. There exists a set that is $\leq_{btt}^p$-complete for $\cc{EXP}$ whose p-isomorphism class is equal to its p-equipollence class. 
\end{corollary}

\begin{proof}
Kurtz, Mahaney, and Royer \cite{KMRCollapse} showed for every $A$ there is a set $B \geq_m^p A$ such that the polynomial-time many-one degree of $B$ was equal to its p-isomorphism class. Such $B$ is necessarily not in $\cc{P}$, since $\cc{P}$ is its own poly-time many-one degree, but contains pairs of non-p-isomorphic sets. So we may apply Observation~\ref{obs:iso}. They similarly showed a collapse result for a $\leq_{btt}^p$-complete set for $\cc{EXP}$, which is not in $\cc{P}$ by the Time Hierarchy Theorem, so again we may apply Observation~\ref{obs:iso}.
\end{proof}

In Proposition~\ref{prop:cardvsiso} we exhibited a p-cardinal (that of $\Sigma^*$) containing infinitely many p-isomorphism classes; we believe much more should be true, in fact, we expect such cardinals are dense in the $\leq_m^p$ ordering:

\begin{conjecture} \label{conj:iso}
For every set $A$, there is a set $B \geq_m^p A$ such that $\card{B}$ contains infinitely many distinct p-isomorphism classes.
\end{conjecture}

We next observe that Berman \& Hartmanis's Cantor--Bernstein-like proof extends to p-cardinality:

\begin{theorem} \label{thm:BH}
Let $A \subseteq \Sigma^*, B \subseteq \Gamma^*$, and let $p\colon A \to B$ and $q\colon B \to A$ be length-increasing, p-invertible partial polynomial-time injective functions. Then $A$ and $B$ have the same p-cardinality.
\end{theorem}

A proof nearly identical to that in \cite[Thm.~1]{BH} works; for expository purposes, we follow the terminology and proof described in \cite[Section~2]{FKR}.

\begin{proof}
Define the following graph $G_{p,q}$ on vertex set $0A \cup 1B$: for every $x \in A$, there is an edge $0x \to 1p(x)$, and for every $y \in B$, there is an edge $1y \to 0q(y)$. The graph is readily seen to be bipartite by construction. Since $p$ and $q$ are functions with $\dom(p) \supseteq A$ and $\dom(q) \supseteq B$, every vertex has out-degree exactly 1. Since $p$ and $q$ are invertible, every vertex has in-degree at most 1. The connected components of $G_{p,q}$ are called chains. We call a vertex a source if it has in-degree 0. Thus the graph is a disjoint union of chains, each of which is one of: finite cycles, two-way infinite paths, a one-way infinite path with its source in $0A$, or a one-way infinite path with its source in $1B$.  Since $p$ and $q$ are length-increasing, there can be no finite cycles. There also cannot be any two-way infinite paths, since in following the chain backwards, the length of the strings would have to be decreasing forever. Since every vertex has out-degree 1, every vertex is part of a chain. 

Let $C_A$ be the vertices that are part of a chain with source in $0A$, and $C_B$ be the vertices in that are part of a chain with source in $1B$. We define
\[
\phi(x) = \begin{cases}
p(x) & \text{ if } 0x \in 0A \cap C_A \\
q^{-1}(x) & \text{ if } 0x \in 0A \cap C_B
\end{cases}
\qquad
\psi(y) = \begin{cases}
p^{-1}(y) & \text{ if } 1y \in 1B \cap C_A \\
q(y) & \text{ if } 1y \in 1B \cap C_B
\end{cases}
\]
As is usual in such proofs, it is readily verified that $\phi$ and $\psi$ are inverses of one another. 

Since $p^{-1}$ need not have domain all of $B$, and $q^{-1}$ need not have domain all of $A$, it remains to verify that $\phi$ (resp., $\psi$) can be computed in $\cc{PF}$ on domain $A$ (resp., $B$). We give the proof for $\phi$, the proof for $\psi$ being similar. Let $t(n)$ be a polynomial such that $p^{-1}$ can be computed in $t(n)$ time on inputs in $A$ and $q^{-1}$ can be computed in $t(n)$ time on inputs in $B$. For inputs in $A$, we modify $q^{-1}$ to a function $\widehat{q^{-1}}$ such that $\widehat{q^{-1}}(x) = q^{-1}(x)$ if $x \in q(B)$, and otherwise $\widehat{q^{-1}}(x)$ outputs a special symbol $\bot$. The latter can be computed in $O(t(n)\log t(n))$ time by running a $t(n)$-time machine for $q^{-1}$, with an additional clock, and if the clock ever hits $t(n)$, then stopping and outputting $\bot$. Similarly we modify $p^{-1}$ to $\widehat{p^{-1}}$. 

The key to computing $\phi$ efficiently is to determine whether $0x \in C_A$ or $0x \in C_B$. It does this by applying $\widehat{q^{-1}}$ and $\widehat{p^{-1}}$ alternately until it gets $\bot$, at which point it has found a source of the chain and thus knows which of $C_A$ or $C_B$ the vertex $0x$ is in. Since $p^{-1}$ and $q^{-1}$ are length-decreasing, this takes no more than time polynomial in $|x|$.
\end{proof}

\begin{remark}
We see from the proof that the main use of the length-increasing condition was that any sequence of $p^{-1}, q^{-1}$ terminated after $\poly(|x|)$ steps. Thus the argument generalizes to any p-well-orderings (see Definition~\ref{def:well}) $\prec_A$ on $A$ and $\prec_B$ on $B$ such that $q(p(x)) \succ_A x$ and $p(q(y)) \succ_B y$ for all $x \in A, y \in B$.

In fact, one can get away with further generalizations. If a sequence of $p^{-1},q^{-1}$ forms a cycle of p-bounded length, this can be detected and easily handled. The real issue with the above proof is cycles or chains such that a polynomial-time algorithm cannot identify when a string is part of such a cycle or chain. (A necessary condition is that such a cycle or chain is super-polynomially long, but there could still be very long such chains/cycles that are easily identifiable for other reasons.)
\end{remark}

\subsection{P-countability: density, rankability, enumerabililty, and compressibility} \label{sec:countable}
It is natural to consider $\Sigma^*$ to be the p-cardinality analogue of $\aleph_0$, the cardinality of $\N$, so here we examine sets that have the same p-cardinality as $\Sigma^*$.

\begin{definition}[p-countable]
A set is \emph{p-countable} if it is finite or has the same p-cardinality as $\Sigma^*$.
\end{definition}

Before we get to interesting connections, we ensure that our exploration is not dependent on alphabet:

\begin{observation}
Let $\Sigma,\Gamma$ be two finite alphabets of size at least 2. Then $\card{\Sigma^*} = \card{\Gamma^*}$.
\end{observation}

\begin{proof}
Let $k = |\Sigma|, \ell = |\Gamma|$. The following is a p-equipollence $\Sigma^* \to \Gamma^*$. Given $x \in \Sigma^*$, consider the natural number $N$ associated to it (see Section~\ref{sec:prelim}), and then let $y$ be the $\ell$-ary string in $\Gamma^*$ associated with $N$, and output $y$. We have that $|x| \sim \log_k N$ and $|y| \sim \log_\ell N$ (since both $k,\ell \geq 2$), so $|y| = \Theta(|x|)$. It is clear that the computations can be done in polynomial time.
\end{proof}

Because every language is a subset of $\Sigma^*$, every $A$ has $\card{A} \preceq \card{\Sigma^*}$, so it would seem that $\Sigma^*$ is the ``largest'' p-cardinal. But remember that $\preceq$ is only a pre-order, not a partial order, and in the proof of Proposition~\ref{prop:not_poset}, we saw that for $B' \not\in \cc{P}$, $\card{\Sigma^*} \prec \card{\Sigma^* \oplus B'}$ (even though we also have $\card{\Sigma^* \oplus B'} \preceq \card{\Sigma^*}$) but the two cardinalities are not equal. So the \emph{$\equiv$-equivalence class} of $\card{\Sigma^*}$ is indeed maximal in the quotient poset, but the p-cardinal itself is not maximal in the $\preceq$ preorder. 

\subsubsection{Enumerability} 
Selman introduced the notion of p-enumerability:

\begin{definition}[{p-enumerable, Selman \cite{SelmanEnum}}]
A language $L$ is \emph{p-enumerable} if there is a function $f \in \cc{FP}$ such that $\img(f) = L$ and such that for all $y \in L$, there exists $x$ such that $f(x)=y$ and $|x| \leq \poly(|y|)$.
\end{definition}

Note that in Selman's definition, $f$ need not be injective.

\begin{observation} \label{obs:enum}
A p-countable set is p-enumerable.
\end{observation}

\begin{proof}
For finite sets this is clear. For an infinite p-countable set $L$, let $f\colon \Sigma^* \to L$ be a p-equipollence. Then $L$ is p-enumerable via $f$; the honesty-like condition is guaranteed by the p-computability of $f^{-1}$.
\end{proof}

This also follows from the fact that the p-enumerable sets are precisely those in $\cc{NP}$ \cite[Cor.~2]{SelmanEnum} and p-countable sets are in $\cc{P}$ (by Proposition~\ref{prop:dtt}). This simple observation will have an interesting consequence when we get to arithmetic of p-cardinals (Proposition~\ref{prop:NPimmune}). 

\begin{definition}[{p-enumerable by iteration \cite{HHSY}}] A language $L$ is \emph{p-enumerable by iteration} if there is a Turing machine $M$ and polynomial $p$ such that for all $x \in L$, $M(x)$ halts within $p(|x|)$ steps, and there is some $x_0 \in L$ such that $L = \{x_0, M(x_0), M(M(x_0)), M(M(M(x_0))), \dotsc\}$. $L$ is \emph{invertibly p-enumerable by iteration} if it is p-enumerable by iteration as before, and the function computed by $M$ has a polynomial-time inverse on its image.
\end{definition}

Hemachandra, Hoene, Siefkes, and Young \cite[Sec.~5]{HHSY}, among other things, gave examples of languages that were invertibly p-enumerate by iteration, including computably enumerable $\cc{P}$-cylinders, and sets with various forms of self-reducibility. Here we add another family of examples:

\begin{proposition} \label{prop:enum}
If $L$ is p-countable then $L$ is invertibly p-enumerable by iteration. The converse does not always hold.
\end{proposition}

\begin{proof}
For finite $L$ this is clear, the enumerator implements a finite cycle through all the elements of $L$.

For infinite $L$, let $f\colon L \to \Sigma^*$ be a p-equipollence. Let $x_0 = f^{-1}(\epsilon)$, and for $x \in L$ define $M(x)$ by $M(x)=f^{-1}(f(x)+1))$. Since $f,f^{-1}$ and addition by constants are all p-invertible, so is $M$. 

Hemachandra \emph{et al.} \cite{HHSY} give examples of languages invertibly p-enumerable by iteration that are not in $\cc{P}$, but by Proposition~\ref{prop:dtt}, p-countable languages are in $\cc{P}$.
\end{proof}

\subsubsection{Compressibility}

Goldsmith, Hemachandra, and Kunen \cite{GHK} defined a language $L$ to be \emph{$\cc{P}$-compressible} if there is an $f \in \cc{FP}$, $f\colon \Sigma^* \to \Sigma^*$ such that $f|_L\colon L \to \Sigma^*$ is a bijection. 

\begin{observation} \label{obs:compressible}
If $L$ is p-countably infinite, then $L$ is p-compressible by a function that is p-invertible.
\end{observation}

\begin{proof}
Given a p-equipollence $f\colon L \to \Sigma^*$, define $\hat{f} \colon \Sigma^* \to \Sigma^*$ by 
\[
\hat{f}(x) = \begin{cases}
x & x \notin L \\
f(x) & x \in L.
\end{cases}
\]
Since $f \in \cc{PF}$ with $f|_L \colon L \to \Sigma^*$ a bijection, and $L \in \cc{P}$ (since it is p-countable), clearly $\hat{f} \in \cc{FP}$ and $\hat{f}|_L \colon L \to \Sigma^*$ is still a bijection. Moreover, $f^{-1} \colon \Sigma^* \to L$ is an inverse of $\hat{f}$ that is in $\cc{FP}$.
\end{proof}

Several results and techniques of Goldsmith, Hemachandra, and Kunen thus have immediate implications for p-countable sets; we record a particularly interesting one here.

\begin{corollary}[{cf. Goldsmith, Hemachandra, Kunen \cite[Prop.~3.18]{GHK}}] \label{cor:GHK}
If $A,\overline{A},B,\overline{B}$ are all p-countably infinite, then $A \cong^p B$. 
\end{corollary}

\begin{proof}
Goldsmith--Hemachandra--Kunen \cite[Prop.~3.18]{GHK} use Grollman \& Selman's result that $\cc{P} = \cc{UP}$ iff one-way functions do not exist \cite{GS}. Here instead, we use the fact that the p-equipollences with $\Sigma^*$ have inverses by definition. But the following construction is otherwise the same as \cite{GHK}. We include it here for completeness and ease of reference.

Let $f_1\colon A \to \Sigma^*$ and $f_2\colon \overline{A} \to \Sigma^*$ be p-equipollences, with extensions $\hat{f}_1,\hat{f}_2$ to total $\cc{FP}$ functions. Define
\[
f(x) = \begin{cases}
0 \hat{f}_1(x) & x \in A \\
1 \hat{f}_2(x) & x \notin A.
\end{cases}
\]
Then we have
\[
f^{-1}(by) = \begin{cases}
f_1^{-1}(y) & b=0 \\
f_2^{-1}(y) & b=1.
\end{cases}
\]
Since $\hat{f}_1, \hat{f}_2, f_1^{-1}, f_2^{-1}$, and the characteristic function of $A$ are all computable in polynomial time (the latter by Proposition~\ref{prop:dtt}), so are $f$ and $f^{-1}$. Furthermore, $f,f^{-1}$ are total bijections $\Sigma^* \to \Sigma^*$. 

Define $g$ similarly for $B$. Then $h(x) = f^{-1}(g(x))$ is a total, p-equipollence $\Sigma^* \to \Sigma^*$ with $h(B) = f^{-1}(g(B)) = f^{-1}(0\hat{g}_1(B)) = f^{-1}(0g_1(B)) = f^{-1}(0\Sigma^*) = f_1^{-1}(\Sigma^*)=A$, so $h$ is the desired p-isomorphism $B \to A$.
\end{proof}

\subsubsection{Rankability} Suppose $L \subseteq \Sigma^*$ is p-countably infinite. Then there is a polynomial-time computable and invertible  bijection $f\colon L \to \Sigma^*$. For infinite languages, ranking functions \cite{GoldbergSipser, AllenderThesis} are a special case of bijective maps to $\Sigma^*$: 

\begin{definition}[{Ranking function, p-rankability \cite{GoldbergSipser, AllenderThesis}}]
The \emph{(strong) ranking function} of a language $L$ is the map $rk_L(x) = |\{y \in L : y \leq_{lex} x\}|$. A language $L$ is \emph{strongly p-rankable} if its strong ranking function can be computed in $\cc{FP}$, and \emph{weakly p-rankable} if the ranking function can be computed in $\cc{PF}$ on those inputs that are actually in $L$. 
\end{definition}

From p-countability, we don't quite get p-rankability, but we get a version if we allow replacing the ordering $\leq_{lex}$ by a different ordering that shares many properties with $\leq_{lex}$. We generalize rankability to arbitrary total orderings $\prec$. We define the \emph{strong ranking function of $L$ with respect to $\prec$} as $rk_{L,\prec}(x) = |\set{y \in L : y \preceq x }|$, and we say $L$ is \emph{strongly (resp., weakly) p-rankable with respect to $\prec$} if $rk_{L,\prec}$ is in $\cc{FP}$ (resp., in $\cc{PF}$ on inputs in $L$).

\begin{definition}[{Length-related \cite{MP79}}]
A total ordering $\prec$ on $\Sigma^*$ is \emph{length-related} if $x \prec y$ implies $|x| \leq \poly(|y|)$.
\end{definition}

\begin{observation}
If $L$ is p-countable, then $L$ is weakly p-rankable under some length-related total (on $L$) ordering $\prec \in \cc{PF}$.
\end{observation}

\begin{proof}
For finite sets this is clear. Let $L$ be an infinite p-countable set, and let $f\colon L \to \Sigma^*$ be a p-equipollence. Define $x \prec y$ by $f(x) <_{lex} f(y)$. It is clear that $\prec$ is computable in $\cc{PF}$ on $L$. Since lexicographic ordering is total on $\Sigma^*$, $\prec$ is total on $L$. Since $f$ is p-invertible, it is honest, so we have $|x| \leq \poly(|f(x)|) \leq \poly(|f(y)|) \leq \poly(|y|)$, where the final inequality follows from the fact that $f$ is computable in polynomial time.
\end{proof}

Next we note that p-cardinality preserves density up to polynomial transformations. Recall that the census function of a language $L$ is $c_L(n) = |L \cap \Sigma^{\leq n}|$. We say that two functions $f,g\colon \N \to \N$ are polynomially related if there are polynomials $p,q$ such that $f(n) \leq g(p(n))$ and $g(n) \leq f(q(n))$ for all $n$.

\begin{lemma} \label{lem:density}
If $\card{A} = \card{B}$, then $c_A(n)$ and $c_B(n)$ are polynomially related.
\end{lemma}

\begin{proof}
Let $f\colon A \to B$ be a p-equipollence witnessing $\card{A} = \card{B}$. Let $p$ be a polynomial bounding the running time of $f$ on inputs in $A$, and similarly let $q$ be a polynomial bound on the running time of $f^{-1}|_B$. Then we have
\[
f(A \cap \Sigma^{\leq n}) \subseteq B \cap \Sigma^{\leq p(n)} \qquad \text{ and } \qquad f^{-1}(B \cap \Sigma^{\leq n} )\subseteq A \cap \Sigma^{\leq q(n)}.
\]
The first implies $c_A(n) \leq c_B(p(n))$ and the second implies $c_B(n) \leq c_A(q(n))$. 
\end{proof}

\begin{corollary}
If $L$ is p-countably infinite, then $L$ is exponentially dense, that is, $L$ has at least $2^{\Omega(n^c)}$ strings of length $\leq n$ for some $c > 0$.
\end{corollary}

Since p-isomorphism preserves p-cardinality, but p-rankability is not p-isomorphism invariant, we recall Goldsmith \& Homer's notion of scalability \cite{GHscalable}. A language is \emph{p-scalable} if it is p-isomorphic to a strongly p-rankable set. (We would like to introduce the obvious notion of weakly p-scalable and say that p-countability implies weak p-scalability, but we still do not know if the latter in fact holds!)

Combining the preceding results we get

\begin{theorem} \label{thm:ranking}
For any exponentially dense language $L$,
\[
\text{strongly p-rankable}  \Longrightarrow \text{p-scalable} \Longrightarrow  \text{p-countable}  \Longrightarrow  \text{weakly p-rankable w.r.t. $\prec$},
\]
for some length-related total ordering $\prec$ on $L$, computable in $\cc{PF}$.
\end{theorem}

So exponential density is necessary to have the same p-cardinality as $\Sigma^*$, and among exponentially dense languages, having the p-cardinality of $\Sigma^*$ sits in between strong and weak p-rankability (for some orderings). We believe the converses of these implications do not hold. 

\begin{proof}
We show the second arrow, the others having been shown above or previously. Goldberg \& Sipser \cite{GoldbergSipser} noted that if a language is strongly p-rankable, then $rk_L^{-1}\colon \Sigma^* \to L$ is computable by binary search in time polynomial in the length of its output.  So it suffices to show that $rk_L$ is honest. By assumption, we have $rk_L(x) \geq 2^{|x|^c}$. Taking bit-lengths of both sides we get $|rk_L(x)| \geq |x|^c -1$, so $|x| \leq |rk_L(x)|^{1/c}$, and thus $rk_L$ is honest.
\end{proof}

\begin{example}
This allows us to give another interesting class of languages exhibiting a difference between p-isomorphism classes and p-cardinalities. Namely, any dense proper subset of $\Sigma^*$ that is p-scalable has the same p-cardinality as $\Sigma^*$, but is not isomorphic to $\Sigma^*$. The reason for the latter is that isomorphic sets have isomorphic complements, but the complement of $\Sigma^*$ is empty, while that of the other language is non-empty by assumption.
\end{example}

\subsubsection{Finite differences.}
We now show that finite differences preserve p-countability (for infinite sets, of course). The first part of this next result, on rankability, is surely known, but we could not find a reference, and its proof serves as a good warm-up for the second part; the second part we believe is new. We use the standard notation $A=^* B$ to mean that $A$ and $B$ differ by at most a finite set.

\begin{theorem} \label{thm:countableDiff}
If $A =^* B$, then
\begin{enumerate}
\item $A$ is strongly (resp., weakly) p-rankable iff $B$ is;

\item $A$ is p-countable iff $B$ is.
\end{enumerate}
\end{theorem}

\begin{proof}
Suppose $A =^* B$. If they are finite they are both strongly p-rankable and both p-countable, so we may now assume they are infinite.

(1) We will define a ``shift function'' $\sigma(x)$ which tells us how much to shift as we move from $A$ to $B$. We define $\sigma$ inductively by
\[
\sigma(x) = \sigma(x-1) + B(x) - A(x)
\]
where $x-1$ denotes the immediate predecessor of $x$ in length-lexicographic order, and in case $x = \epsilon$ (the empty string), we simply define by convention $\sigma(``\epsilon-1") := 0$. 

Write $B = (A \cup P) \backslash N$ where $P$ is a finite set disjoint from $A$, and $N$ is a finite subset of $A$. Let $\{z_1, \dotsc, z_{s+t}\} = P \cup N$ with $z_1 \leq_{lex} z_2 \leq_{lex} \dotsb \leq_{lex} z_{s+t}$. First, $\sigma$ is constant on each range $[\epsilon, z_1), [z_1, z_2), \dotsc, [z_{s+t-1}, z_{s+t}), [z_{s+t}, \infty)$. Next,  $\sigma$ can be computed in polynomial time: the $z_i$ and the values of $\sigma$ on each of the preceding ranges can be hardcoded, so that on input $x$ an algorithm merely decides which of these $O(1)$ ranges $x$ is in, and then outputs the right (hard-coded) value. By construction, we have $rk_B(x) = rk_A(x) + \sigma(x)$ for all strings $x$, giving the first statement.

(2) Suppose $A$ is p-countably infinite (the finite case was handled in the first paragraph), and $A=^* B$. Let $f \colon \N \to A$ be a p-equipollence. We will define a p-equipollence $g\colon \N \to B$ by shifting $f$, similar to the above. Let $B = (A \cup P) \backslash N$, where $P$ is a finite set disjoint from $A$ and $N$ is a finite subset of $A$. 

First we show that $B' = A \cup P$ is p-countable. The idea is to first enumerate $P$, then shift up the enumeration of $A$ by $|P|$. More formally, we define $g'(0), g'(1), \dotsc, g'(|P|-1)$ to bijectively enumerate $P$, and then $g'(i) = f(i-s)$ for all $i \geq s$. It is readily verified that $g'\colon \N \to B'$ is a p-equipollence (to see p-computability and p-invertibility, note that those first $|P|$ values can be hard-coded).

Next, we show that $B = B' \backslash N$ is p-countable. Let $N = \{g'(j_1), \dotsc, g'(j_t)\}$ with $j_1 < j_2 < \dotsb < j_t$. The idea is to use the same p-equipollence $g'$, but to shift it at appropriate points to omit $g'(j_1), \dotsc, g'(j_t)$. Let $h \colon \N \to \N \backslash \{j_1, \dotsc, j_t\}$ be the bijection such that $h(i)$ is the $i$-th smallest element of $\N \backslash \{j_1, \dotsc, j_t\}$. It is readily verified that $g = g' \circ h$ is an equipollence $\N \to B$. To see that it is p-computable and p-invertible, it suffices to show these for $h$, since $g'$ is already a p-equipollence. For both of these, note that $h$ is piecewise linear with only finitely many breakpoints. By hard-coding those breakpoints and their jump values (NB: it could jump more than one, if some of the $j_k$ are consecutive integers), $h$ is p-computable. Since $h$ is also strictly increasing, computing its inverse amounts to figuring out which of the intervals of linearity the input is in, and then inverting the linear function on that interval. Again, as there are only finitely many such intervals, the necessary information can be hard-coded and the remaining computation is then easy.
\end{proof}

This raises the question of whether, for infinite sets, two sets that are finitely different must always be p-equipollent. We thank an anonymous reviewer for a suggested construction, a slight modification of which led us to the following resolution of this question in the negative:

\begin{theorem} \label{thm:diff}
There is an infinite language $A \in \cc{P}$ such that
\begin{enumerate}
\item $A$ is not p-equipollent with any of its proper subsets (it is Dedekind p-finite, see Definition~\ref{def:Dedekind}); 

\item For any two languages $B,B'$ that are both finitely different from $A$, we have $\card{B}=\card{B'}$ if and only if $|B \backslash A| - |A \backslash B| = |B' \backslash A| - |A \backslash B'|$. In particular, not all sets finitely different from $A$ are p-equipollent to $A$; and

\item For any two languages $B,B'$ that are both finitely different from $A$, we have $\card{B} \preceq \card{B'}$ iff $|B \backslash A| - |A \backslash B| \leq |B' \backslash A| - |A \backslash B'|$.
\end{enumerate}
\end{theorem}

In particular, for this $A$, the preorder $\preceq$ on the p-cardinals of sets finitely different from $A$ is in fact a total order of the same order type as the integers. Although part (1) of the theorem already answers the question of whether finite differences preserve p-cardinality, the other two parts follow a similar proof, and will have an interesting consequence below (Corollary~\ref{cor:counterex}). 

\begin{proof}
Let $A = \{1^{2^{2^{2^{k}}}} : k \in \N \}$. Note that if $\ell = 2^{2^{2^k}}$ is the length of a string in $A$, the next longest string has length exactly $\ell^{\log_2 \ell}$ (this is why we used a tower of three exponentials; if we only use a tower of two exponentials, instead of $\ell^{\log \ell}$, we get $\ell^2$ here).

(1) Suppose $f \colon A \to A'$ is a p-equipollence to a proper subset $A' \subsetneq A$. Let $n_0$ be such that for all $n \geq n_0$, the running time of $f$ on strings of length $\ell$ is strictly less than $\ell^{\log \ell}$. Then $f$ cannot map any string in $A$ of length at most $n_0$ to a string of length $> n_0$, so we must have that $f$ gives a bijection between $A \cap \Sigma^{\leq n_0}$ and $A' \cap \Sigma^{\leq n_0}$, and since $A' \subseteq A$, in fact $A'$ must agree with $A$ on $\Sigma^{\leq n_0}$. 

But on strings of length $\geq n_0$, $f$ must be length non-increasing, since for such a string, $f(x)$ does not have time to write out any longer string in $A$. But since $f$ was already a bijection from the strings in $A$ of length $\leq n_0$ to themselves, any $x \in A$ of length $\geq n_0$ has nowhere to go but itself. Thus on strings of length $\geq n_0$, $f$ must be the identity. Therefore $A' = A$.

(2) Let $B,B'$ be finitely different from $A$. 

($\Rightarrow$) Suppose $\card{B} = \card{B'}$; let $f \colon B \to B'$ be a p-equipollence. Let $n_0$ denote one more than the maximum length of any string $x$ such that $B(x) \neq A(x)$ or $B'(x) \neq A(x)$. Note that, above length $n_0$, $B,B'$ agree with $A$, so they have the same $\ell^{\log \ell}$ gap between their elements. 
Thus, as above, there exists some $n_1 \geq n_0$ such that neither $f$ nor $f^{-1}$ can map any string of length $\leq n_1$ to a string of length $> n_1$. But then $f$ must map $B \cap \Sigma^{\leq n_1}$ bijectively to $B' \cap \Sigma^{\leq n_1}$ (with inverse $f^{-1}$ restricted to $B' \cap \Sigma^{\leq n_1}$. 

Now, let $P = B \backslash A, N = A \backslash B$, and similarly define $P',N'$ with $B'$ in place of $B$. Then $f|_{\Sigma^{\leq n_1}}$ gives a bijection from $[(A \cup P) \backslash N] \cap \Sigma^{\leq n_1}$ to $[(A \cup P') \backslash N'] \cap \Sigma^{\leq n_1}$. As $P,N,P',N' \subseteq \Sigma^{\leq n_1}$---as part of the definition of $n_1$---and these are finite sets, and $P,P'$ are disjoint from $A$ and $N,N'$ are subsets of $A$, we necessarily have $|P| - |N| = |P'| - |N'|$.

($\Leftarrow$) Let $n_0, P, P', N, N'$ be as above, and suppose $|P|-|N| = |P'|-|N'|$. The following map $f$ is a p-equipollence $B \to B'$. On strings of length $\leq n_0$, $f$ implements an arbitrary, hard-coded bijection from $[(A \cup P) \backslash N] \cap \Sigma^{\leq n_0}$ to $[(A \cup P') \backslash N'] \cap \Sigma^{\leq n_0}$, as these are both finite sets of the same cardinality. On strings of length $> n_0$, $f$ is the identity.

(3) Similar to (2). For the ($\Rightarrow$) direction, let $f \colon B \to B'' \subseteq B'$ be a p-equipollence. As above, $f$ must be the identity map for all sufficiently long strings in $B$, say above some length $n_1$. But then $f$ gives a bijection $B \cap \Sigma^{\leq n_1} \to B'' \cap \Sigma^{\leq n_1} \subseteq B \cap \Sigma^{\leq n_1}$. As before, counting these finite sets then tells us that $|B \backslash A| - |A \backslash B| \leq |B' \backslash A| - |A \backslash B'|$. For the ($\Leftarrow$) direction, a construction similar to part (2) works, where the hard-coded part of $f$ is now an injection rather than a bijection.
\end{proof}

\begin{question} \label{q:finite}
If $A,B$ are infinite languages and $A=^* B$ with at least $n^c$ strings of length $\leq n$ (for some $c > 1$ and all sufficiently large $n$), must $\card{A} = \card{B}$? What if additionally $A \subsetneq B$ with $B \backslash A$ finite? What if, additionally, $A \in \cc{P}$?
\end{question}

Actually, the set in Theorem~\ref{thm:diff} does not even have logarithmic density!

\subsection{Arithmetic of p-cardinals}
Since many versions of AC are about cardinal arithmetic, we develop arithmetic of p-cardinals.

\begin{proposition}[p-cardinal arithmetic] \label{prop:arithmetic}
If $\card{A} = \card{A'}$ and $\card{B} = \card{B'}$, then $\card{A \oplus B} = \card{A' \oplus B'}$ and $\card{A \times B} = \card{A' \times B'}$. If $\card{A} \preceq \card{A'}$ and $\card{B} \preceq \card{B'}$, then $\card{A \oplus B} \preceq \card{A' \oplus B'}$ and $\card{A \times B} \preceq \card{A' \times B'}$. 

Defining $\card{A} + \card{B} = \card{A \oplus B}$ and $\card{A} \times \card{B} := \card{A \times B}$, the p-cardinals form a pre-ordered semi-ring with additive identity $\card{\emptyset}$ and multiplicative identity $\card{\{\epsilon\}}$.

Furthermore:
\begin{enumerate}
\item For all $A$, $\card{\emptyset} \times \card{A} = \card{\emptyset}$.

\item For all $A$, $\card{\emptyset} \preceq \card{A} \preceq \card{\Sigma^*}$.

\item If $A$ is infinite, then $\card{A} \succ \card{F}$ for any finite set $F$.

\item For all $A,B$, $\card{A} \preceq \card{A} + \card{B}$.

\item For all $A$ and for all nonempty $B$, $\card{A} \preceq \card{A} \times \card{B}$.

\item The p-cardinalities of finite sets form a totally ordered sub-semiring that is isomorphic as an ordered semiring to $(\N,+,\times,\leq,0,1)$. 

\item For all finite non-empty languages $n$ and all languages $A$, we have $\card{n \times A} = \card{A \oplus A \oplus \dotsb \oplus A}$, where the latter summation has $|n|$ summands.

\end{enumerate}

\end{proposition}

\begin{proof}
Suppose $f\colon A \to A''$ is a p-equipollence for some $A'' \subseteq A'$ and $g\colon B \to B''$ is a p-equipollence for some $B'' \subseteq B'$. Then $f \oplus g$, defined by 
\[
(f \oplus g)(bx)  =  \begin{cases}
f(x) & \text{ if } b = 0 \\
g(x) & \text{ if } b = 1
\end{cases}
\]
is a p-equipollence $A \oplus B \to A'' \oplus B'' \subseteq A' \oplus B'$, with inverse $f^{-1} \oplus g^{-1}$. Similarly, $f \times g$, defined by $(f \times g)(x,y) = (f(x), g(y))$, is a p-equipollence $A \times B \to A'' \times B'' \subseteq A' \times B'$ with inverse $f^{-1} \times g^{-1}$. This gives us the desired statements for $\preceq$; since $\preceq$ is only a pre-order, the statements for $=$ do not immediately follow. However, if in the proof we take $A'' = A'$ and $B''=B'$ then we get the desired statements for $=$.

Finally, note that $A \oplus \emptyset = \{0x : x \in A\}$, which is easily seen to have the same p-cardinality as $A$. Similarly, $A \times \{b \} = \{(a,b) : a \in A\}$ is again easily seen have the same p-cardinality as $A$. 

We leave the further properties as exercises for the reader.
\end{proof}

\begin{remark} \label{rmk:equiv}
Note that the arithmetic operations and comparisons are also well-defined on the $\equiv$-equivalence classes of p-cardinals (where $\card{A} \equiv \card{B}$ if $\card{A} \preceq \card{B}$ and $\card{B} \preceq \card{A}$), in which case the result is a partially ordered semi-ring, not just pre-ordered. 
\end{remark}

\begin{observation}
$\card{\Sigma^*} = \card{\Sigma^*} + \card{\Sigma^*}$ and $\card{\Sigma^*} = \card{\Sigma^*} \times \card{\Sigma^*}$. 
\end{observation}

\begin{proof}
Let $n_x$ denote the natural number associated to $x$ (see Preliminaries), and let $\sigma_n$ denote the string associated to the natural number $n$. 
The bijection $\Sigma^* \to \Sigma^* \oplus \Sigma^*$ is given by 
\[
f(x) = \begin{cases}
0\sigma_{n_x/2} & \text{ if $n_x$ is even} \\
1\sigma_{(n_x-1)/2} & \text{ if $n_x$ is odd}
\end{cases}
\]
with inverse 
\[
f^{-1}(by) = \begin{cases}
\sigma_{2n_y} & \text{ if } b = 0 \\
\sigma_{2n_y+1} & \text{ if } b = 1
\end{cases}.
\]

For the direct product, we use any of the standard bijections $\Sigma^* \to \Sigma^* \times \Sigma^*$ (for example, using the bijections $n,\sigma$ above to translate to $\N$ and using the standard bijections $\N \to \N \times \N$). 
\end{proof}

\begin{corollary} \label{cor:cylinder}
If $A$ is p-countable, then
\[
A \cong^p \overline{A} \qquad \Longleftrightarrow \qquad \overline{A} \text{ is p-countable},
\]
and if these hold, then $A \cong^p A \times \Sigma^*$ (i.\,e., $A$ is a p-cylinder).
\end{corollary}

\begin{question} \label{q:cylinder}
What can be said if merely $\card{A} = \card{A \times \Sigma^*}$ and $\card{\overline{A}} = \card{\overline{A} \times \Sigma^*}$?
\end{question}

\begin{proof}
Suppose $A$ is p-countable. The forward direction follows from the fact that p-isomorphism preserves p-countability. The backward direction follows from Corollary~\ref{cor:GHK} using $A=B$.

For the final statement, suppose that both $A,\overline{A}$ are p-countable. The idea of the proof is similar to that of Goldsmith, Kunen, and Hemaspaandra \cite{GHK} and Corollary~\ref{cor:GHK}. Since $\card{A} = \card{\Sigma^*}$ and  $\card{\Sigma^*} = \card{\Sigma^*} \times \card{\Sigma^*}$, we may replace the LHS and one of the ones on the RHS with $A$ to get $\card{A} = \card{A} \times \card{\Sigma^*} = \card{A \times \Sigma^*}$. Similarly for $\overline{A}$. Now let $f\colon A \to A \times \Sigma^*$ and $g\colon \overline{A} \to \overline{A} \times \Sigma^*$ be p-equipollences. Then the function
\[
h(x) = \begin{cases}
f(x) & x \in A \\
g(x) & x \notin A.
\end{cases}
\]
is a p-computable bijection $\Sigma^* \to \Sigma^* \times \Sigma^*$, since $A \in \cc{P}$. To see that it is a bijection, note that $\overline{A \times \Sigma^*} = \overline{A} \times \Sigma^*$. The inverse of $h$ is given similarly by the p-computable bijection
\[
h^{-1}(y) = \begin{cases}
f^{-1}(y) & y \in A \times \Sigma^* \\
g^{-1}(y) & y \notin A \times \Sigma^*.
\end{cases}
\]
Thus $A \cong A \times \Sigma^*$, so $A$ is a cylinder.
\end{proof}

We note that the set of \emph{all} p-cardinalities is uncountably infinite: there are uncountably many subsets of $\Sigma^*$, and each p-cardinality class has countably many representatives, since there are only countably many Turing machines (even ones that are only partial) to compute p-equipollences. However, even in interesting countable sets of p-cardinalities, we believe:

\begin{conjecture} \label{conj:not_fin_gen}
For any standard complexity class $\mathcal{C} \supseteq \cc{P}$ (such as $\cc{P}, \cc{NP}, \cc{PSPACE}, \cc{EXP}$, etc.), the semiring of p-cardinalities of languages in $\mathcal{C}$ is not finitely generated. The same for the semiring of equivalence class of p-cardinalities (see Remark~\ref{rmk:equiv}).
\end{conjecture}

For ordinary cardinals, $\aleph_0$ is the smallest infinite cardinal, and every infinite cardinal $\mathfrak{c}$ satisfies $\mathfrak{c} + \aleph_0 = \mathfrak{c}$. However, for p-cardinals this is no longer the case:

\begin{proposition} \label{prop:NPimmune}
If $A$ is $\cc{P}$-immune, then $\card{A} \not\succeq \card{A} + \card{\Sigma^*}$.
\end{proposition}

\begin{proof}
Suppose $\card{A} + \card{\Sigma^*} \preceq \card{A}$. Since the left-hand side is equal to $\card{A \oplus \Sigma^*}$, we have a p-equipollence $f$ from $A \oplus \Sigma^*$ to a subset $A' \subseteq A$. Then the restriction of $f$ to $1\Sigma^*$ gives a p-equipollence to a subset $A'' \subseteq A' \subseteq A$. Thus $A''$ is p-countable, thus by Proposition~\ref{prop:dtt} is in $\cc{P}$. So $A$ contains an infinite $\cc{P}$ subset, contradicting its $\cc{P}$-immunity.
\end{proof}

\begin{proposition}[{Non-cancellation of addition, Nerode \& Remmel \cite[Thm.~7]{NerodeRemmelPET}}]
There exist languages $A,B,C$ with $\card{A \oplus B} = \card{A \oplus C}$ but $\card{B} \neq \card{C}$.
\end{proposition}

Nerode \& Remmel showed the result for unary languages, and using slightly different definitions, but their same construction works with our definitions.

Nerode \& Remmel \cite[Thm.~10]{NerodeRemmelPET} also showed that one does get cancellation of multiplication by finite  sets of unary strings. Their argument used a back-and-forth construction that required them to compute their function on all predecessors of a string. Because they were using unary languages, there were only linearly many such predecessors and the entire argument could be carried out in polynomial time. When we attempt to do the same in our setting with, say, length-lexicographic ordering, a string has exponentially many predecessors (exponential in its length), so the same strategy doesn't work. It is possible that such a strategy could work with a p-well-founded ordering (see Definition~\ref{def:well}) in which immediate predecessors were computable in $\cc{FP}$, but we have not been able to get this to work, so we leave it as a question:

\begin{question}[Division by finite p-cardinals] \label{q:division}
Is it the case that for any finite non-empty language $n$ and for any languages $A,B$, $\card{n \times A} = \card{n \times B}$ implies $\card{A} = \card{B}$?
\end{question}

\subsection{Polynomial-time Axioms of Choice based on p-cardinality}
\begin{cform}[{See \cite[CN6 and CN7, pp. 52--53]{RubinRubin}}]
For any cardinals $m < n$ and $p < q$, $m + p < n + q$ and $mp < nq$.
\end{cform}

Compare to Proposition~\ref{prop:arithmetic}.

\begin{cform}[{Tarski's Theorem \cite{Tarski}, see \cite[CN 3, p.~52]{RubinRubin}}] For every infinite set $A$, there is a bijection between $A$ and $A \times A$.   \end{cform}

\begin{pform} For every infinite language $L \in \cc{P}$, $\card{L} = \card{L \times L}$. \end{pform}

\begin{cform}[{Law of Trichotomy \cite[T, p.~9]{RubinRubin}}] For all sets $A,B$, either $A$ is in bijection with a subset of $B$ or $B$ is in bijection with a subset of $A$. \end{cform}

A straightforward polynomial-time version of this is false by density considerations:

\begin{proposition} \label{prop:not_total}
There are two languages in $\cc{P}$ whose p-cardinalities are not comparable.
\end{proposition}

\begin{proof}
Let $tow(0) = 1$ and $tow(n+1) = 2^{tow(n)}$ be the tower function. $A_0$ will consists of all strings $x$ of length $tow(2k) \leq |x| < tow(2k+1)$ for all $k \geq 0$, and $A_1$ will consist of all strings of length $tow(2k+1) \leq |x| < tow(2k+2)$ for all $k \geq 0$. Clearly both $A_0, A_1$ are in $\cc{P}$. The idea is that their densities fluctuate between $\Theta(n)$ and $2^{\Theta(n)}$ infinitely often, but at opposite times. Let us verify this. 

For lengths in the range $[\frac{1}{2}tow(2k+2), tow(2k+2))$, all the strings in $A_0$ have length $< tow(2k+1)$, so there are at most $2^{tow(2k+1)+1}-1 \sim tow(2k+2)$ of them. So in these ranges we have $c_{A_0}(n) \in [n/2, n]$. In this same range, all strings are present in $A_1$, so we have $c_{A_1}(n) \geq 2^n$, which is nearly maximal. For lengths in the range $[\frac{1}{2}tow(2k+1), tow(2k+1))$ the situation is reversed.

For the sake of contradiction, suppose $B \subseteq A_0$ has $\card{B} = \card{A_1}$. Since $A_0$ is empty at lengths in the ranges $[tow(2k+1), tow(2k+2))$, so is $B$, and thus $B$ has density at most $O(n)$ for $n$ in the range $[\frac{1}{2}tow(2k+2), tow(2k+2)]$. But $A_1$ has density at least $2^n$ in this range, so the two are not polynomially related for arbitrarily large stretches of lengths, a contradiction. The analogous argument swapping the roles of $A_0$ and $A_1$ and using the ranges $[tow(2k), tow(2k+1))$ instead rules out a subset of $A_1$ having the same p-cardinality as $A_0$. Thus $\card{A_0}$ and $\card{A_1}$ are incomparable.
\end{proof}

We thus propose a polynomial-time form of AC with density in its assumption.

\begin{pform}
For all sets $A,B \in \cc{P}$ of polynomially related densities, either $\card{A} \preceq \card{B}$ or $\card{B} \preceq \card{A}$.
\end{pform}

If we remove the restriction that both languages are in $\cc{P}$, we believe the statement becomes false. 

\begin{cform}[{Law of Trichotomy \cite[T', p.~10]{RubinRubin}}] For every two non-empty sets, there is a surjective mapping of one onto the other. \end{cform}

\begin{pform}
For every two non-empty languages $A,B$ with $c_A(n) \leq c_B(\poly(n))$, there is an honest surjective polynomial-time function $A \to B$.
\end{pform}

Classically, T implies T'. For the p-analogues, such an implication seems tantamount to Hypothesis Q, though formalizing this has proved tricky. We do have, however:

\begin{cform} Every surjective function has a right inverse. \end{cform}

\begin{pform} \label{pac:Q} Every polynomial-time computable, honest, surjective function has a polynomial-time inverse.  This is an exact restatement of Hypothesis Q. \end{pform}

\begin{cform}[{See \cite[Sec.~6, pp. 52--53]{RubinRubin}}]
Cardinal forms of the axiom of choice. Throughout, $m,n,p,q$ denote infinite cardinals. When not quantified, universal quantification is assumed, e.\,g., the first condition is more precisely ``For all infinite cardinals $m,n$, $m \cdot n = m + n$.'' 
\begin{enumerate} 
\item $m \cdot n = m + n$.

\item There is a cardinal $n$ such that $m=n^2$.

\item If $m^2 = n^2$ then $m=n$.

\item If $m + p < n + p$ then $m < n$.

\item If $mp < np$ then $m < n$.

\item Every cardinal has an immediate successor ($m < n$ and if $m < p$ then $n \leq p$).

\item If $n < p$ and there is no cardinal between $n$ and $p$ ($p$ ``covers'' $n$) then either $mn = mp$ or $mp$ covers $mn$. (If $p$ is an immediate successor of $n$ then $p$ covers $n$. The converse is equiavlent to AC.)

\item If $m < n$ then there is a $p$ such that $n=mp$.

\item If $m < n$ then $n/m$ exists (as in the previous point) and is unique.

\item If $m < n$ then $n/m = n$.

\item If $m + p = m + q$ then either $p=q$, or $p \leq m$ and $q \leq m$.

\item If $m + m < m + n$ then $m < n$. (The converse is independent of AC.)

\item If $m < n$ then $n-m$ exists, that is, there exist one and only one p such that $n=m+p$. (Existence alone is independent of AC.)

\item If $m < n$ then $n-m=n$.

\item If $p < n, q < n$ then $p + q \neq n$.

\item If $p < n, q < n$ then $pq \neq n$.

\item Either $mn = m$ or $mn= n$.

\end{enumerate}
\end{cform}

Many of the arguments that these are equivalent to one another and to the standard AC rely on the notion of the immediate successor of a cardinal (e.\,g., \cite{RubinRubin}). We show that such a construction is unlikely to exist for (infinite) p-cardinalities, by a Ladner-type diagonalization result.

Before coming to the result, our proof currently uses one additional assumption, essentially because of the issue raised by Theorem~\ref{thm:diff}.

\begin{definition}[Much smaller p-cardinality]
We say that $\card{A}$ is \emph{much smaller} than $\card{B}$, denoted $\card{A} \ll \card{B}$, if $\card{A} \prec \card{B}$ and for any subset $B' \subseteq B$ such that $\card{A} = \card{B'}$, the difference $B \backslash B'$ is infinite.
\end{definition}

The next proposition shows that $\ll$ is in fact a well-defined relationship on p-cardinality classes.

\begin{proposition} \label{prop:muchless}
If $\card{A} = \card{\hat{A}}$ and $\card{B} = \card{\hat{B}}$, then $\card{A} \ll \card{B}$ iff $\card{\hat{A}} \ll \card{\hat{B}}$.
\end{proposition}

\begin{proof}
Suppose $\card{A} \ll \card{B}$, we show that $\card{\hat{A}} \ll \card{\hat{B}}$, with the reverse implication following by symmetry. If $\card{A} = \card{\hat{A}}$, then for any subset $B' \subseteq B$, we have $\card{A} = \card{B'}$ iff $\card{\hat{A}} = \card{B'}$, so we have $\card{\hat{A}} \ll \card{B}$. Now suppose $\card{B} = \card{\hat{B}}$, and let $\hat{B}' \subseteq \hat{B}$ have the same p-cardinality as $\hat{A}$. Let $f\colon \hat{B} \to B$ be a p-equipollence, and define $B' = f(\hat{B}')$. Then we have $\card{A} = \card{B'}$, and by assumption $B \backslash B'$ is infinite. But then $f^{-1}(B \backslash B') = \hat{B} \backslash \hat{B}'$ is infinite as well, since $f^{-1}$ is injective on $B$.
\end{proof}

Obviously finite sets are much smaller than $\card{A}$ for any infinite $A$, but we also have less trivial examples. 

\begin{proposition}
1.If $A \subset B$ and $A \not\equiv_m^p B$, then $\card{A} \ll \card{B}$.

2. If $A \not\equiv_{m}^p B \neq \Sigma^*$, then $\card{A} \ll \card{\Sigma^* \oplus B}$.
\end{proposition}

\begin{proof}
1. If $A \subset B$ then $\card{A} \preceq \card{B}$. If $A \not\equiv_m^p B$, suppose $B' \subset B$ has the same p-cardinality as $A$. Since $A,B'$ are both proper subsets of $B$, neither can be $\Sigma^*$, so Proposition~\ref{prop:dtt} implies $A \equiv_m^p B'$. But since $B \not\equiv_m^p A$, we must have $B \not\equiv_m^p B'$, and thus $B \backslash B'$ must be infinite.

2. Define $\hat{B} = \Sigma^* \oplus B$ and $\hat{A} = 0A$. As $\hat{B} \equiv_m^p B$ and $\hat{A} \equiv_m^p A$, we still have $\hat{A} \not\equiv_m^p \hat{B}$, and now we have $\hat{A} \subseteq \hat{B}$. Part 1 implies $\card{\hat{A}} \ll \card{\hat{B}}$, and Proposition~\ref{prop:muchless} then gives $\card{A} \ll \card{\hat{B}}$ as well.
\end{proof}

\textbf{Warning!} ``Much smaller'' is, unfortunately, not transitive. Suppose $A \not\equiv_m^ p B \neq \Sigma^*$. Applying part 2 of the preceding proposition a few times, we get $\card{A} \ll \card{\Sigma^* \oplus B} \ll \card{\Sigma^* \oplus A} \ll \card{\Sigma^* \oplus \Sigma^* \oplus B} = \card{\Sigma^* \oplus B}$. But obviously $\card{\Sigma^* \oplus B}$ is equal to itself, so it cannot be much less than itself. At first sight this seems to point to some error, but it is actually a natural consequence of the fact that p-cardinality mixes subsets and many-one degrees, while the poset of subset inclusion is not compatible with the poset of many-one reductions. Viz., there are sequences of languages $L_1 \subset L_2 \subset L_3$ with $L_2$ $\cc{NP}$-complete but $L_1, L_3 \in \cc{P}$, e.\,g., $0L_1 \subseteq \Sigma^* \oplus SAT \subseteq \Sigma^*$ for any $L_1 \in \cc{P}$. 

In contrast, the relation ``$\preceq$ but not $\ll$'' (``smaller, but not much smaller'') is transitive. For if $A$ is p-equipollent to a cofinite subset of $B$, and $B$ is p-equipollent to a cofinite subset of $C$, then composing the two gives a p-equipollence between $A$ and a cofinite subset of $C$. We will use this fact (in contrapositive form) several times in the following proof.

\begin{theorem} \label{thm:cardinalDensity}
Let $A,B$ be languages with $\card{A} \ll \card{B}$. Then there exists an infinite language $C$ with $\card{A} \prec \card{C} \prec \card{B}$. 
\end{theorem}

Because of Theorem~\ref{thm:diff}, the assumption here is not equivalent to $\card{A} \prec \card{B}$; we will see in Corollary~\ref{cor:counterex} below that the same set constructed in Theorem~\ref{thm:diff} can be used to give an example of sets $A,B$ with $\card{A} \prec \card{B}$ but with no p-cardinal strictly in between the two. In the vein of Question~\ref{q:finite}, it is interesting to ask for conditions under which Theorem~\ref{thm:cardinalDensity} would hold with the weaker assumption $\card{A} \prec \card{B}$, and/or with the stronger conclusion $\card{A} \ll \card{C} \ll \card{B}$. As is, we necessarily have that at least one of $\card{A} \prec \card{C}$ and $\card{C} \prec \card{B}$ can be turned into $\ll$, for otherwise we could not have $\card{A} \ll \card{B}$, but our construction does not control which one.

\begin{proof}
By assumption, there is a p-computable, p-invertible bijection from $A$ to a co-infinite subset of $B$, but not vice versa. (Note that $B$ cannot be finite.) To simplify notation in what follows, let us identify $A$ with its image under this bijection, so that, without loss of generality, we may assume that in fact $A \subsetneq B$ and $B \backslash A$ is infinite. We will build $C$ such that $A \subsetneq C \subsetneq B$ and $\card{A} \prec \card{C} \prec \card{B}$. Since $A \subseteq C \subseteq B$, we have $\card{A} \preceq \card{C} \preceq \card{B}$. So it suffices to build $C$ in between $A$ and $B$ avoiding p-computable, p-invertible bijections $C \to A$ and $B \to C$. 

Let $\hat{M}_1, \hat{M}_2, \hat{M}_3, \dotsc$ be an enumeration of Turing machine transducers, and for all $c$ let $M_c$ be the machine gotten from $\hat{M}_c$ by restricting $\hat{M}_c$ to run on each input of length $n$ for no more than $cn^c + c$ steps. If $\hat{M}_c(x)$ has not made an output after $c|x|^c + c$ steps, then $M_c(x)$ rejects $x$ (i.\,e., it makes no output). It is clear that $M_1, M_2, \dotsc$ is thus an enumeration of all partial polynomial-time Turing machine transducers. 

\textbf{Construction.} 
We build $C$ in stages $s=0,1,2,\dotsc$. We also build a set $E$ of strings to be excluded from $C$. For each $s \in \N$, $C_s$ and $E_s$ will be the parts of $C$ and $E$, respectively, enumerated at the end of stage $s$. The construction will guarantee that, for each $s$, $C_s \cap E_s = \emptyset$, $A \subseteq C_s \subseteq C_{s+1} \subseteq B$, $C_s$ will be finitely different from $A$, $E_s \subseteq E_{s+1} \subseteq B$, and $E_s$ will be finite.

We will have two sets of requirements to satisfy below; when $s \equiv 0 \pmod{3}$, we will add to $C_s$ to ensure that $C$ is infinite. The stages $s+1$ where $s \equiv 1 \pmod{3}$ (resp. $2 \pmod{3}$) will ensure the first (resp., second) set of requirements are satisfied at this stage. We break the construction into numbered cases for reference in the verification below.

\textit{Stage $s=0$.} Set $C_0 = A$ and $E_0 = \emptyset$.

\textit{Stage $s+1$, where $s \equiv 0 \pmod{3}$.} Let $y$ be the least element of $B \backslash (C_s \cup E_s)$ (any element will do), and set $C_{s+1} := C_s \cup \{y\}$ and $E_{s+1} := E_s$. (Such a $y$ must exist since $B \backslash A$ is infinite, $C_s$ is finitely different from $A$, and $E_s$ is finite.)

\textit{Stage $s+1$, where $s = 3 \langle \alpha, \beta \rangle + 1$.} (Here we use $\langle \bullet, \bullet \rangle$ to denote a bijection $\N \times \N \to \N$.) 

\begin{enumerate}[label=Case \arabic{enumi}:]
\item $M_\alpha|_{A} \colon A \to M_{\alpha}(A)$ is not a p-equipollence, or $M_\alpha(A) \not\subseteq B$, or $M_\beta$ is not its inverse p-equipollence $M_\alpha(A) \to A$. In this case, set $C_{s+1} := C_s$ and $E_{s+1} := E_s$, and continue to the next stage ($s+2$).

\item Otherwise. Note that in this case $B \backslash M_\alpha(A)$ must be infinite, since $\card{A} = \card{M_\alpha(A)}$, but $\card{A} \ll \card{B}$. Let $Y = B \backslash E_s$; since $E_s$ is finite, $Y$ is cofinite in $B$. Since $\card{A} \ll \card{B}$ by assumption, but $Y$ is cofinite in $B$, $\card{A}$ cannot be equal to $\card{Y}$.

\begin{enumerate}[label=Subcase 2.\arabic{enumii}:]
\item $Y \not\subseteq \dom(M_\beta)$. In this subcase, let $y$ be the least element of $Y \backslash \dom(M_\beta)$ (any element will do), set $C_{s+1} := C_s \cup \{y\}$, $E_{s+1} := E_s$, and go to the next stage.

\item Not in the previous subcase, and $M_\beta$ is not injective on $Y$. In this subcase, let $y_1, y_2 \in Y$ be the least two distinct elements that $M_\beta$ maps to the same place (any two will do). Set $C_{s+1} := C_s \cup \{y_1, y_2\}$ and $E_{s+1} := E_s$ and continue to the next stage.

\item Not in the previous subcases, and $M_\beta(Y) \not\subseteq A$. In this subcase, let $y$ be the least element of $Y$ such that $M_\beta(y) \notin A$ (any will do), set $C_{s+1} := C_s \cup \{y\}$, $E_{s+1} := E_s$, and go to the next stage.

\item Not in the previous subcases. In this subcase, there must exist $y \in Y$ such that $M_\alpha(M_\beta(y)) \neq y$; for at this point we have that $M_\beta|_Y$ is an injective function from $Y$ into $A$, so if $M_\alpha$ were its inverse we would have $\card{A} = \card{Y}$ (for we already know that $(M_\beta \circ M_\alpha)|_{A} = id_A$), which we showed above was impossible. Let $y$ be the least such (any such will do), set $C_{s+1} := C_s \cup \{y\}$, $E_{s+1} := E_{s}$, and go to the next stage.
\end{enumerate}
\end{enumerate}

\textit{Stage $s+1$ where $s = 3 \langle \alpha, \beta \rangle + 2$.} 
\begin{enumerate}[label=Case \arabic{enumi}:]
\item In any of the following subcases, set $C_{s+1} := C_s$, $E_{s+1} := E_s$, and go to the next stage.

\begin{enumerate}[label=Subcase 1.\arabic{enumii}:]
\item $M_\alpha(C_s) \not\subseteq B$.

\item It is not the case that $M_\alpha|_{C_s} \colon C_s \to M_\alpha(C_s)$ is a p-equipollence with inverse $M_\beta$.

\item $X := B \backslash (M_\alpha(C_s) \cup E_s)$ is not contained in $\dom(M_\beta)$.

\item Not in the above subcases, and $M_\beta(X) \subseteq C_s$.
\end{enumerate}

\item Not in any of the above subcases. In this case, let $x$ be the least element of $X$ such that $M_\beta(x) \notin C_s$ (any such will do), set $E_{s+1} := E_s \cup \{x\}$, $C_{s+1} := C_s$, and go to the next stage. (Note that such an $x$ must exist, for $E_s$ is finite, and $M_\alpha$ is a p-equipollence from a set that is finitely different from $A$ to a subset of $B$, but $\card{A} \ll \card{B}$.)
\end{enumerate}

%
%
%
%

This completes the description of stage $s+1$, and thus of the construction. 

\textbf{Requirements.} Below, we will show that, in addition to ensuring $C$ is infinite, the construction satisfies two sets of requirements, and that these requirements imply the desired property of $C$. For the purposes of stating these requirements, define $E = \bigcup_s E_s = \lim_{s \to \infty} E_s$. 
The requirements for $(\alpha,\beta)$ are as follows:

\begin{quotation}
R1$_{\alpha, \beta}$: At least one of the following holds\\
(a) $A \not\subseteq \dom(M_\alpha)$ or $M_\alpha$ is not injective on $A$, or \\
(b) $C \not\subseteq \dom(M_\beta)$ or $M_\beta$ is not injective on $C$, or \\
(c) $M_\alpha(A)$ 
is not contained in $B$ (sic!), or \\
(d) $M_\beta(C)$ is not contained in $A$, or \\
(e) there is $a \in A$ with $M_\beta(M_\alpha(a)) \neq a$ or there is $c \in C$ with $M_\alpha(M_\beta(c)) \neq c$. 

R2$_{\alpha, \beta}$: At least one of the following holds\\
(a) $C \not\subseteq \dom(M_\alpha)$ or $M_\alpha$ is not injective on $C$, or \\
(b) $B \not\subseteq \dom(M_\beta)$ or $M_\beta$ is not injective on $B$, or \\
(c) $M_\alpha(C)$ is not contained in $B$, or \\
(d) $M_\beta(B) \cap E$ is nonempty, or \\ 
(e) there is $c \in C$ with $M_\beta(M_\alpha(c)) \neq c$ or there is $b \in B$ with $M_\alpha(M_\beta(b)) \neq b$.
\end{quotation}

\textbf{The requirements suffice.} First let us see that the requirements, plus $C$ being infinite, suffice for the theorem. The requirement R1$_{\alpha,\beta}$ implies that $M_\alpha$ is not a p-equipollence $A \to C$ with inverse $M_\beta$, and R2$_{\alpha,\beta}$ implies that $M_\alpha$ is not a p-equipollence $C \to B$ with inverse $M_\beta$ (note that R2$_{\alpha,\beta}$(d) ensures $M_\beta(B) \not\subseteq C$, since $E$ is disjoint from $C$). Since $(\alpha,\beta)$ run over a complete list of polynomial-time Turing machines, these requirements will establish that $\card{A} \prec \card{C} \prec \card{B}$. 

The stages with $s+1 \equiv 0 \pmod{3}$ add elements to $C$ infinitely often, ensuring that $C$ is infinite even if $A$ was finite.

\textbf{Verification of the requirements.} We begin by noting that this is an ``injury-free'' argument, in the sense that once a requirement is satisfied, it never becomes unsatisfied (``is never injured''). This is because each requirement is a $(\Sigma^0_1)^{A \oplus B}$ statement, that is, it is of the form $(\exists \vec{x})[P(\vec{x})]$ where $P$ is a predicate that is computable with an oracle for $A$ and $B$. Thus, it suffices to show that at the end of stage $s+1 = 3 \langle \alpha,\beta \rangle + b$, if $C',E'$ are any two sets satisyfing $C_s \subseteq C' \subseteq B$, $E_s \subseteq E' \subseteq B \backslash C'$, then R$b_{\alpha,\beta}$ is satisfied with $C'$ (resp., $E'$) in place of $C$ (resp., $E$).
 
\begin{remark} \label{rmk:requirements}
To ensure this $(\Sigma^0_1)^{A \oplus B}$ property, we note that R1$_{\alpha,\beta}$ is slightly stronger than is needed to ensure $M_\alpha$ is not a p-equipollence $A \to C$ with inverse $M_\beta$, and similarly R2$_{\alpha,\beta}$ is slightly stronger than is needed to ensure $M_\alpha$ is not a p-equipollence $C \to B$ with inverse $M_\beta$. In particular, if R1$_{\alpha,\beta}$(c) were merely ``$M_\alpha(A) \neq C$'' and R1$_{\alpha,\beta}$(d) were ``$M_\beta(C) \neq A$,'' it would suffice to avoid p-equipollences, but then these statements would not be $\Sigma^0_1$ in $A \oplus B$, for it would then be possible that later additions to $C$ would violate such a condition. Similarly, if R2$_{\alpha,\beta}$(c) were merely ``$M_\alpha(C) \neq B$'' and R2$_{\alpha,\beta}$(d) were merely ``$M_\beta(B) \neq C$,'' it would suffice to avoid p-equipollences, but again, would not so clearly avoid injuries from future additions to $C$.
\end{remark}

\textit{Stage $s+1$ with $s = 3\langle \alpha, \beta\rangle + 1$ ensures R1$_{\alpha,\beta}$ is satisfied.} We break the verification into cases according to the cases in the construction.

\begin{enumerate}[label=Case \arabic{enumi}.]
\item If $M_\alpha|_A \colon A \to M_\alpha(A)$ is not a p-equipollence to a subset $M_\alpha(A) \subseteq B$ with inverse $M_\beta$, then Case 1 of the construction tells us to move onto the next stage. We must thus show that in this case, without any further changes to $C_s, E_s$, it is already the case that R1 is satisfied. 

In this case, one of the following must hold:
\begin{enumerate}[label=(\Alph{enumii})]
\item $A \not \subseteq \dom(M_\alpha)$ or $M_\alpha$ is not injective on $A$. 

\item $M_\alpha(A) \not \subseteq B$.

\item $M_\alpha(A) \not\subseteq \dom(M_\beta)$

\item $M_\beta$ is not injective on $M_\alpha(A)$

\item $M_\beta(M_\alpha(A)) \not\subseteq B$

\item $M_\alpha$ and $M_\beta$ are not inverses on $A$, $M_\alpha(A)$, respectively. That is, there is some $a \in A$ such that $M_\beta(M_\alpha(a)) \neq a$ or some $b \in M_\alpha(A)$ such that $M_\alpha(M_\beta(b)) \neq b$. The latter is equivalent to the existence of an $a \in A$ such that $M_\alpha(M_\beta(M_\alpha(a))) \neq M_\alpha(a)$.
\end{enumerate}

Suppose (A) holds. (A) is the same as R1$_{\alpha,\beta}$(a), so R1$_{\alpha,\beta}$ is satisfied.

Suppose (B) holds. (B) is the same as R1$_{\alpha,\beta}$(c).

Suppose (C) holds. Then there is an $a \in A$ such that $M_\alpha(a) \notin \dom(M_\beta)$ and thus $M_\beta(M_\alpha(a)) \neq a$ since the LHS is undefined, so R1$_{\alpha,\beta}$(e) holds.

Suppose (D) holds. Then there are distinct $a_1,a_2 \in A$ such that $M_\beta(M_\alpha(a_1)) = M_\beta(M_\alpha(a_2))$. But then at least one of $M_\beta(M_\alpha(a_i)) = a_i$ for $i=1,2$ must fail, so R1$_{\alpha,\beta}$(e) holds.

Suppose (E) holds. Since $M_\beta(M_\alpha(A)) \not\subseteq B$, there is some $a \in A$ such that $M_\beta(M_\alpha(a)) \notin B \supseteq A$, so we must have $M_\beta(M_\alpha(a)) \neq a$, so R1$_{\alpha,\beta}$(e) holds.

Suppose (F) holds. Then it must be the case that there is an $a \in A$ such that $M_\beta(M_\alpha(a)) \neq a$. For if not, then by the second part of (F) we get that there is an $a$ such that 
\[
M_\alpha(M_\beta(M_\alpha(a))) \neq M_\alpha(a).
\]
But if $M_\beta(M_\alpha(a)) = a$ for all $a \in A$, then the displayed equation simplifies to $M_\alpha(a) \neq M_\alpha(a)$, which is absurd. Thus R1$_{\alpha,\beta}$(e) is satisfied. 

Thus, case 1 ensures R1$_{\alpha,\beta}$, as claimed.

\item Suppose instead we are in case 2, in which $M_\alpha|_A \colon A \to M_\alpha(A)$ is a p-equipollence to a subset of $B$ with inverse $M_\beta$. Following the notation in the construction, let $Y = B \backslash E_s$.

\begin{enumerate}[label=Subcase 2.\arabic{enumii}.]
\item  If the condition of subcase 2.1 is satsified ($Y \not \subseteq \dom(M_\beta)$), then we add a $y \in Y \backslash \dom(M_\beta)$ to $C$, thus satisfying R1$_{\alpha,\beta}$(b).

\item In subcase 2.2, if the relevant condition is satisfied, we add to $C$ distinct $y_1, y_2 \in Y$ such that $M_\beta(y_1) = M_\beta(y_2)$. This satisfies the second part of R1$_{\alpha,\beta}$(b).

\item In subcase 2.3, if the relevant condition is satisfied, we add to $C$ a $y \in Y$ such that $M_\beta(y) \notin A$, thus satisfying R1$_{\alpha,\beta}$(d).

\item In subcase 2.4, if the relevant condition is satisfied, we add to $C$ a $y \in Y$ such that $M_\alpha(M_\beta(y)) \neq y$, thus satisfying the second part of R1$_{\alpha,\beta}$(e). 
\end{enumerate}
\end{enumerate}

Since the cases and subcases exhaust all possibilities (case 2 is the ``otherwise'' of case 1, and subcase 2.4 captures all cases not in the previous cases by construction), this stage ensures that R1$_{\alpha,\beta}$ is satisfied.

\textit{Stage $s+1 = 3\langle \alpha, \beta \rangle + 2$ ensures R2$_{\alpha,\beta}$.} Following the notation in the construction, let $X = B \backslash (M_\alpha(C_s) \cup E_s)$. 

\begin{enumerate}[label=Case \arabic{enumi}.]
\item We must show that in this case, without any further changes to $C_s, E_s$, it is already the case that R2 is satisfied. 

\begin{enumerate}[label=Subcase 1.\arabic{enumii}:]
\item $M_\alpha(C_s) \not\subseteq B$. Thus we have that $M_\alpha(C) \not\subseteq B$, which is precisely R2$_{\alpha,\beta}$(c).

\item $M_\alpha|_{C_s} \colon C_s \to M_\alpha(C_s)$ is not a p-equipollence with inverse $M_\beta$. In this case, one of R2$_{\alpha,\beta}$(a), (b), or (e) must be satisfied (since together they are the definition of a p-equipollence), and thus R2$_{\alpha,\beta}$ is satisfied already.

\item $X \not \subseteq \dom(M_\beta)$. Since $X \subseteq B$ by definition, R2$_{\alpha,\beta}$(b) is satisfied.

\item Not in the above subcases, and $M_\beta(X) \subseteq C_s$. In this case, we claim that $M_\beta$ cannot be injective on $B$. To see that $M_\beta$ is not injective on $B$, first note that since we are not in subcase 1.2, we have that $M_\beta$ is the p-equipollence $M_\alpha(C_s) \to C_s$ that is inverse to $M_\alpha|_{C_s}$. Thus $M_\beta$ maps $M_\alpha(C_s)$ bijectively onto $C_s$. Now let $x \in X$, and let $b = M_\alpha(M_\beta(x))$. Since $M_\beta(X) \subseteq C_s$ by assumption, we have $b \in M_\alpha(C_s)$. But $X$ is disjoint from $M_\alpha(C_s)$ by definition, so $b \neq x$. But since $M_\beta, M_\alpha$ are inverses on $M_\alpha(C_s), C_s$ respectively, and $M_\beta(X) \subseteq C_s$ by assumption, we have $M_\beta(b) = M_\beta(M_\alpha(M_\beta(x))) = (M_\beta \circ M_\alpha)(M_\beta(x)) = (id)(M_\beta(x)) = M_\beta(x)$. Thus $M_\beta$ is not injective on $B$, satisfying the second part of R2$_{\alpha,\beta}$(b).
\end{enumerate}
\item In this case, $E_s$ gets extended to $E_{s+1}$ by adding some $x \in X$ such that $M_\beta(x) \notin C_s$. Since $X \subseteq B$, this ensures that $M_\beta(B) \cap E$ is not empty, satisfying R2$_{\alpha,\beta}$(d). 
\end{enumerate}

Since the cases and subcases exhaust all possibilities, at the end of this stage the requirements R2$_{\alpha,\beta}$ are satisfied.

Putting these all together, the construction indeed produces an infinite set $C$ such that $\card{A} \prec \card{C} \prec \card{B}$, completing the proof of the theorem.
\end{proof}

\begin{remark}
In addition to R1 and R2 being stronger than needed (see Remark~\ref{rmk:requirements}), the reader may also wonder why R1 and R2 are not perfect analogues of one another. The answer is that we actually first wrote a proof with them being more symmetric, then pared them down to the only parts of the requirements we actually ended up needing. 

Finally, we note that the argument is \emph{almost} monotone in $C$. The only requirement that is not monotone in $C$ is R2$_{\alpha,\beta}$(d), and it was this requirement that forced us to keep track of the set $E$ of excluded strings (R2$_{\alpha,\beta}$(d) is monotone in $E$). 
\end{remark}

Given that $\aleph_0$ is the unique minimum infinite cardinal, it is natural to wonder whether minimal infinite p-cardinals exist. The following corollary answers this in the negative:

\begin{corollary}[No minimal infinite p-cardinals] \label{cor:minimal}
For every infinite $B \subseteq \Sigma^*$, there is an infinite set $C$ such that $\card{C} \prec \card{B}$.
\end{corollary}

\begin{proof}
Apply Theorem~\ref{thm:cardinalDensity} with $A = \emptyset$. 
\end{proof}

We suspect a modification of the above proof could be used to answer the following stronger question in the negative:

\begin{question} \label{q:minimal}
Do there exist p-cardinals that are minimal(ly infinite) under $\equiv$? That is, an infinite language $B \subseteq \Sigma^*$ such that if $\card{C} \prec \card{B}$, then either $C$ is finite or $\card{B} \preceq \card{C}$?
\end{question}

We conclude this section by showing that the assumption $\card{A} \ll \card{B}$ in Theorem~\ref{thm:cardinalDensity} cannot, in general, be replaced by $\card{A} \prec \card{B}$, even for infinite languages (for finite languages this is clear, whenever $|B|=|A|+1$), by the following corollary to Theorem~\ref{thm:diff}:

\begin{corollary}[Infinite p-cardinals can have immediate predecessors] \label{cor:counterex}
There exist infinite languages $A,B$ with $\card{A} \prec \card{B}$ such that any language $C$ with $\card{A} \preceq \card{C} \preceq \card{B}$ is p-equipollent to either $A$ or $B$.
\end{corollary}

By Theorem~\ref{thm:cardinalDensity}, necessarily any such pair cannot have $\card{A} \ll \card{B}$.

\begin{proof}
Let $A$ be the set guaranteed by Theorem~\ref{thm:diff}, let $w \notin A$, and let $B = A \cup \{w\}$. Suppose $C$ is such that $\card{A} \preceq \card{C} \preceq \card{B}$, but $C$ is not p-equipollent to $B$. We will show that $C$ must be p-equipollent to $A$. Since $\card{C} \preceq \card{B}$ by assumption, there is a p-equipollence $f \colon C \to C'$ to a proper subset $C' \subset B$. Since neither the conclusion nor the statement depends on the representative of $\card{C}$ chosen, without loss of generality let us replace $C$ by $C'$. 

Since $B = A \cup \{w\}$ and $C' \subseteq B$, $C'$ is either a subset of $A$, or is $A' \cup \{w\}$ for a subset $A' \subseteq A$. But by assumption, there is a p-equipollence from $A$ to a subset of $C'$. By Theorem~\ref{thm:diff}(1), $C'$ cannot be a proper subset of $A$. If $C'=A$, then we are done. 

So all that is left is to handle the case that $C' = A' \cup \{w\}$ for some $A' \subseteq A$. If $A' = A$, then $C' = B$, contradicting our assumption that $\card{C} \neq \card{B}$. Thus $A'$ must be a proper subset of $A$. 
If $|A \backslash A'| \geq 2$, then 
let $w_1 \in A \backslash A'$. Then we have $\card{C'} = \card{A' \cup \{w\}} = \card{A' \cup \{w_1\}}$. But $A' \cup \{w_1\}$ is a proper subset of $A$, contradicting Theorem~\ref{thm:diff}(1). So $|A \backslash A'| \leq 2$, and since $A'$ is a proper subset of $A$, we must have $|A \backslash A'| = 1$. Thus $|C' \backslash A| - |A \backslash C'| = |\{w\}| - |A \backslash A'| = 1-1=0$, 
so by Theorem~\ref{thm:diff}(2), $\card{C'} = \card{A}$.
\end{proof}

\section{Statements weaker than the Axiom of Choice} \label{sec:weaker}
Here we consider statements that are implied by the Axiom of Choice over ZF, but that are not known to be equivalent to AC (or are known to be strictly weaker), because their polynomial-time analogues make interesting connections.

\begin{cresult} Any union of countably many countable sets is countable. \end{cresult}

\begin{panalogue} \label{pta:union} If $L$ is a $\cc{P}$ collection of languages $L_x$ each of which is p-countable, then $L$ is p-countable. \end{panalogue}

One might hope to build a p-equipollence $L \to \Sigma^*$ by using the individual p-equipollences $L_x \to \Sigma^*$ together with the fact that $\card{\Sigma^*} = \card{\Sigma^* \times \Sigma^*}$. However, one runs into the issue that if $f_x\colon L_x \to \Sigma^*$ is a p-equipollence, there may not be a uniform polynomial upper bound on the runtimes of all the $f_x$. In the opposite direction, the fact that $L$ itself is in $\cc{P}$ means there \emph{is} a uniform polynomial upper bound on the times to decide each $L_x$, so perhaps there is still some hope that this is unconditionally true. We leave it as an open question

\begin{question} \label{q:union}
Prove or disprove the PTA\ref{pta:union}. Does one of the polynomial-time versions of AC imply PTA\ref{pta:union}?
\end{question}

\begin{cresult} If $A$ is infinite, there is an injection from $\N \to A$. \end{cresult}

A natural p-analogue of this fails:

\begin{corollary}
There is an infinite language $L \in \cc{P}$ such that $\card{\Sigma^*} \not\preceq \card{L}$.
\end{corollary}

This is a corollary to the proof of Proposition~\ref{prop:not_total}.

\begin{proof}
Let $L$ be one of the languages constructed in Proposition~\ref{prop:not_total}. $\card{\Sigma^*} \preceq L$ would imply $L$ is exponentially dense, but $L$ is infinitely-often linearly sparse, so the result follows.
\end{proof}

\begin{cresult}[The Axiom of Uniformization] If $R \subset X \times Y$ where $X$ and $Y$ are Polish spaces, then there is a subset $f \subseteq R$ that is a partial function $f\colon X \to Y$ and such that $\dom(f) = \set{x \in X | \exists y \in Y (x,y) \in R}$. \end{cresult}

\begin{panalogue} If $X, Y \in \cc{P}$ and $R \subseteq X \times Y$ is in $\cc{P}$ as well, then there is a refinement $f$ of $R$ such that $f$ is a partial function that is polynomial-time computable and whose domain is exactly $\dom(f) = \set{x \in X : \exists y \in Y (x,y) \in R}$.  \end{panalogue}

This is a restatement of $\cc{NPMV}_g \subseteq_c \cc{PF}$, studied by Blass \& Gurevich \cite{blassGurevich2} about a decade before Selman introduced the $\cc{NPMV}$ notation \cite{SelmanFunctions}. Note that some of our polynomial-time versions of AC are about $\cc{NPMV}_t$ or $\cc{NPMV}_{gt}$, but collapses of $\cc{NPMV}_t$ are generally not known to imply collapses of $\cc{NPMV}$, in contrast to the classical, unbounded situation where AC implies the Axiom of Uniformization.

\section{Discussion and future work} \label{sec:discussion}
Beyond the many questions and conjectures already stated (Questions~\ref{q:NPUP}, \ref{q:pairwise_disjoint},  \ref{q:finite}, \ref{q:cylinder}, \ref{q:division}, \ref{q:minimal}, \ref{q:union}, questions around Theorem~\ref{thm:ranking}, Conjectures~\ref{conj:iso} and \ref{conj:not_fin_gen}), here we raise a few more that we think would be interesting avenues of exploration. 

\textit{Around the Isomorphism Conjecture.} Are Joseph \& Young's creative sets \cite{JY}, proposed as counterexamples to the Isomorphism Conjecture \cite{BH}, p-equipollent to each other? To SAT? (See, e.\,g., the introduction of \cite{KMRRandom} for a survey of developments around this.) Recall that a language $A$ is a \emph{cylinder} if $A \cong^p A \times \Sigma^*$; this is equivalent to $A$ being \emph{paddable}, and cylinders play an important role more generally in investigations around the Isomorphism Conjecture \cite{BH}. One of the main theorems of Berman \& Hartmanis \cite{BH} is that an $\cc{NP}$-complete set is p-isomorphic to SAT iff it is a cylinder. Note that if $A$ is a cylinder then $\card{A} = \card{A \times \Sigma^*} \succeq \card{\Sigma^*}$ Note that $A$ is a cylinder iff $\overline{A}$ is. What can be said about the cardinalities of cylinders? If $A$ and $\overline{A}$ both have p-cardinalities equal to their product with $\Sigma^*$, must $A$ be cylinder-like in some way (cf. Question~\ref{q:cylinder})?

\textit{Connections to other complexity notions.} Can some of our open questions about p-cardinality or polynomial-time axioms of choice be resolved assuming $\mu_p(\cc{NP}) \neq 0$ \cite{Lutz}? 

It seems there should be interesting connections to be had between p-cardinality and several of Selman's other interests, such as p-selectivity \cite{Selman79Psel, Selman81Psel, SelmanPsel, NaikSelmanPsel, HT}, disjoint pairs \cite{GSSZdisjoint, GSTWdisjoint, GHSWdisjoint, GSSdisjoint}, and mitoticity \cite{GPSZmitosis, GOPS07, GNSW17}. 

Several versions of AC use well-orderings. In the context of studying self-reducibility, Meyer \& Paterson introduced some definitions that could be useful for polynomial-time analogues of these versions of AC:

\begin{definition}[{Meyer \& Paterson \cite[Def.~4]{MP79}}] \label{def:well}
A partial order $\prec$ on $\Sigma^*$ is \emph{polynomially well-founded} if there is a polynomial $p$ such that every finite $\prec$-decreasing chain has at most $p(|x|)$ elements in it, where $|x|$ is the maximal length of any element in the chain.

A language $L \subseteq \Sigma^*$ is \emph{self-reducible} if there is a polynomial-time oracle TM $M^{\square}$ and an p-well-founded, length-related ordering $\prec$ such that $L(M^L)=L$ and for any input $x$, $M^L(x)$ only queries the oracle on strings that are strictly $\prec x$. 
\end{definition}

Selman studied this notion of self-reducibility in (at least) \cite{SelmanSelfReducible, HNOSPsel, SelmanPsel}.

\begin{cform}[See {\cite[WE1, p.~1]{RubinRubin}}] Every set can be well-ordered. \end{cform}

\begin{pform} Every language in $\cc{P}$ admits a p-well-ordering computable in $\cc{P}$. \end{pform}

More generally, aside from the use of well-orderings in relation to AC, an anonymous reviewer has suggested, and we agree, that it would be interesting to develop a theory of p-\emph{ordinals}, similar to the development of p-cardinals here. We point out that this could either be in terms of total orders on $\Sigma^*$ computable in $\cc{FP}$, or in terms of total orders on subsets $A \subseteq \Sigma^*$, where the order relation is computable in $\cc{PF}$ on $A \times A$. It is unclear to us at this point which would fit more naturally with our development of p-cardinals here.

Recall that a set is said to be Dedekind finite iff it does not have a bijection onto any of its proper subsets. In classical set theory, ``finite'' and ``Dedekind finite'' are equivalent, but in other settings this need not be the case. Dekker \cite{Dekker} and Dekker--Myhill \cite{DekkerMyhill} were originating works on isols as models of Dedekind-finite sets.

\begin{definition}[Dedekind p-finite] \label{def:Dedekind} A language $L \subseteq \Sigma^*$ is \emph{Dedekind p-finite} if it does not have the same p-cardinality as any of its proper subsets. \end{definition}

The set constructed in Theorem~\ref{thm:diff} is an infinite set that is Dedekind p-finite. How do Dedekind p-finite sets relate to the p-finite sets of Nerode \& Remmel \cite[Def.~6]{NerodeRemmelIsol}? What can Dedekind p-finite sets tell us about p-cardinality more generally? Do they have a nice theory of arithmetic? Can they be used to realize Downey's suggestion \cite{DowneyReview} (see the quote in Section~\ref{sec:intro})? 

Lastly, one could also consider finite axioms of choice, in which one starts with a collection consisting of \emph{finite} sets---here one may think correspondingly of a language $L$, where a single element of $L$ represents a binary set by considering the binary encoding of $x$ as the characteristic function of a finite subset of $\mathbb{N}$. Finite AC are discussed in \cite[Sec.~2.2]{Herrlich} and \cite{Truss}.

\section*{Acknowledgment}
We thank Lance Fortnow for initial discussions and some preliminary results on these problems when we first discussed them (I was in graduate school), and also for suggesting I submit it in memoriam. We thank two anonymous reviewers for their suggestions, which both improved the exposition and organization of the paper, and added results to it. In particular, Corollary~\ref{cor:minimal} was originally posed as an open question, but because of expository improvements to the proof of Theorem~\ref{thm:cardinalDensity} suggested by the reviewers, it needed only a 1-line modification to the original proof of that theorem (the theorem in the original submission did not guarantee that $C$ was infinite). Theorem~\ref{thm:diff} answered an open question posed in the original submission, using a modification of a construction suggested by one of the reviewers.

The preparation of this work was partially supported by NSF CAREER award CCF-2047756, and probably one of Lance's awards from when I was in graduate school, but that was long enough ago we don't remember and the NSF probably no longer cares. 

And of course we thank Alan for his immense curiosity and insight that he shared so well with us all.

\appendix

\section{Comments on other versions of the Axiom of Choice}
In this appendix we gather some versions of AC we came across whose polynomial-time versions seem feasible and interesting to study, but we have not undertaken such explorations. 

\subsection{Algebraic versions}
\begin{cform} Every vector space has a basis. \end{cform}

\cite{NerodeRemmel1} give oracles $A_1,A_2$ relative to which $\cc{P}^{A_i} \neq \cc{NP}^{A_i}$, and relative to $A_1$ every non-deterministic polynomial-time subspace has a polynomial-time basis, while relative to $A_2$ there are non-deterministic polynomial-time subspaces with no polynomial-time bases. But, as with the rest of their line of work, here vectors are all encoded over a unary alphabet.

\begin{cform} For every non-empty set $S$ there is a binary operation defined on $S$ that makes $S$ a group. \end{cform}

The polynomial-time version of this seems related to the notions of groupy witnesses \cite[Sec.~4.1.2]{FortnowGrochow} and group definability \cite{AV}.

\begin{cform}[{Tychonoff's Theorem, see \cite[P5, p.~69]{RubinRubin}}] The product of compact spaces is compact. \end{cform}

\subsection{Model-theoretic versions}

\begin{cresult} Every game $G_{S}$ in which $S$ is a Borel subset of Baire space is determined. \end{cresult} 

\begin{cform}[{See \cite[P6, p.~69]{RubinRubin}}] A formula having a model in a set of cardinality $n$ also has a model in a set of cardinality $m$ if $\aleph_0 \leq m \leq n$. \end{cform}

\begin{cform}[{See \cite[P7, p.~69]{RubinRubin}}] A formula having a model in a set of cardinality $\aleph_0$ also has a model in a set of any cardinality greater than $\aleph_0$. \end{cform}

\begin{cform}[{See \cite[P8, p.~69]{RubinRubin}}] If $Q$ is a set of formulas in which the set of individual constants has cardinality $m$ and every finite subset of $Q$ has a model, then $Q$ has a model in a set whose cardinality is not greater than $m + \aleph_0$. \end{cform}

\subsection{Order-theoretic versions}

\begin{cform}[{Hausdorff Maximum Principle, see \cite[Thm.~2.2]{Herrlich}}] In any poset, every totally ordered subset is contained in a maximal totally ordered subset.  Equivalently: every poset (merely) \emph{has} a maximal totally ordered subset. \end{cform}

\begin{cform}[{Kurepa's Maximal Antichain Condition, see \cite[Thm.~2.4(3)]{Herrlich}}] Every poset has a maximal antichain. \end{cform}

Recall that a collection $\mathcal{S}$ of sets is \emph{of finite character} if for every set $A \in \mathcal{S}$, every finite subset $F \subseteq A$ is also in $\mathcal{S}$, and conversely, that is, if there exists a set $A$ such that every finite subset $F \subseteq A$ is in $\mathcal{S}$, then $A$ is also in $\mathcal{S}$.

\begin{cform}[{Tukey's Lemma, see \cite[Thm.~2.2]{Herrlich}, \cite[M7, p.~13]{RubinRubin}}] Every non-empty collection of finite character has a maximal element with respect to inclusion.  \end{cform}

\begin{pform}
Every $\cc{P}$ collection of honestly non-empty languages of finite character has a maximal element with respect to inclusion.
\end{pform}

Finally, we mention a few based on cardinal exponentiation. Cardinal exponentiation is usually defined as $|A|^{|B|} := |\{f\colon A \to B\}|$. 

\begin{cform}[{See \cite[CN23, CN24, CN25, pp.~53--54]{RubinRubin}}]
\begin{enumerate}
\item For each fixed infinite cardinal $m$, for all infinite cardinals $p,q$, if $p^m < q^m$ then $p < q$.

\item There is a cardinal $n > 1$ such that for all cardinals $p,q$ there is a cardinal $1 < m \leq n$ such that $p^m < q^m$ implies $p < q$.

\item For all $m,p,q$, if $m^p < m^q$ and $m \neq 0$ then $p < q$.
\end{enumerate}
\end{cform}

On the one hand, when $B$ is finite, we should have $A^B = A \times A \times \dotsb \times A$, where the latter product has $|B|$ copies of $A$. On the other hand, when $B$ is infinite, it feels natural to potentially define the p-exponent $A^B$ as the set of Turing machines (/programs/indices) that compute partial polynomial-time functions $B \to A$. This causes two issues: first, when $B$ is finite, these two definitions disagree. Second, using the latter definition, whenever $A,B$ are nonempty, $A^B$ is uncomputable by Rice's Theorem. We nonetheless make the following provisional definition:

\begin{definition}[Exponentiation of p-cardinals]
Let $A \subseteq \Sigma^*, B \subseteq \Gamma^*$. We define $\card{A}^{\card{B}}$ as the p-cardinality of the set of Turing machines $M$ such that $M$ computes a polynomial-time partial function $B \to A$. 
\end{definition}

Let us unwrap what it means for expressions such as $\card{A}^{\card{C}} \preceq \card{B}^{\card{C}}$. This means there is a partial polynomial-time computable, p-invertible function $f$ (invertible on $\img(f)$) that, given any Turing machine $M$ computing a polynomial-time function $C \to A$, outputs a Turing machine $f(M)$ computing a partial polynomial-time function $C \to B$. This gives us some hope that, despite the uncomputability of $\card{A}^\card{B}$, meaningful statements might be provable about it.

\begin{pform}
\begin{enumerate}
\item For each fixed $C \in \cc{P}$, for all languages $A,B \in \cc{P}$, if $\card{A}^{\card{C}} \prec \card{B}^{\card{C}}$, then $\card{A} \prec \card{B}$.

\item There is an set $N \in \cc{P}$ with $\card{N} \succ 1$ such that for all sets $A,B \in \cc{P}$, there is a set $1 \prec \card{C} \preceq \card{N}$ in $\cc{P}$ such that $\card{A}^{\card{C}} \prec \card{B}^{\card{C}}$ implies $\card{A} \prec \card{B}$.

\item For all nonempty $C \in \cc{P}$ and all languages $A,B \in \cc{P}$, if $\card{C}^{\card{A}} \prec \card{C}^{\card{B}}$, then $\card{A} \prec \card{B}$.
\end{enumerate}
\end{pform}

\subsection{Forms of AC which it was difficult to formulate an interesting polynomial-time analogue}
Here we catalog some of the versions of the Axiom of Choice for which we had difficulty formulating a reasonable polynomial-time analogue, and the difficulties we encountered.

The first few make reference to the power set, which caused us trouble:

\begin{cform} \label{ac:subset} For any set $A$, its power set (with the empty set removed) has a choice function.  Equivalently: for any set $A$ there is a function $f$ such that for any nonempty subset $B \subseteq A$, $f(B) \in B$. \end{cform}

\begin{cform}[{See \cite[Thm.~2.4(5)]{Herrlich}}] The power set of each well-orderable set can be well-ordered. \end{cform}

On the one hand, we would like to talk about the ``polynomial-time power set'' of a language $L \in \cc{P}$. One attempt is to consider the collection of all $L' \subseteq L$ such that $L' \in \cc{P}$; a natural way to refer to these subsets $L'$ is to write $L' = L \cap L(M_{x})$ for some polynomial-time machine $M_{x}$, and then these indices $x$ can be used as the input to a choice function or well-ordering function. But then we run into the issue that any such polynomial-time function cannot even simulate all $M_x$, given $x$ as input.

Others made reference to the existence of a single element, which felt not amenable to asymptotic considerations, e.\,g., perhaps one of the most famous:

\begin{cform}[{Zorn's Lemma, see, e.\,g., \cite[Thm.~2.2]{Herrlich}}] Every non-empty poset in which every chain (totally ordered subset) has an upper bound contains at least one maximal element. \end{cform}

\bibliographystyle{plainurl}
\bibliography{AC}

\end{document}